\documentclass{article}
\usepackage[letterpaper, portrait, margin=1in]{geometry}

\usepackage{natbib}
\usepackage[ruled]{algorithm2e}

\usepackage{pifont} 
\newcommand{\cmark}{\ding{51}}%
\newcommand{\xmark}{\ding{55}}%

\usepackage{subfigure}
\usepackage{amsmath}
\usepackage{amsthm}
\usepackage{amssymb}
\usepackage{graphicx}

\usepackage{multirow}

\usepackage{color}

\newcommand\domain{\mathcal{X}}

\newcommand\dss{DSS}

\newtheorem{theorem}{Theorem}
\newtheorem{lemma}{Lemma}
\theoremstyle{definition}
\newtheorem{definition}{Definition}

\usepackage[table, usenames, dvipsnames]{xcolor}

\definecolor{ikgreen}{RGB}{200,255,200}
\definecolor{ikred}{RGB}{255,200,200}

\definecolor{mymaroon}{HTML}{800000}
\definecolor{comgray}{HTML}{BEBEBE}
\definecolor{citegreen}{HTML}{458B00}
\usepackage{hyperref}
\hypersetup{
   colorlinks=false,
   citecolor=citegreen
}

\usepackage[ruled]{algorithm2e}

\SetAlFnt{\small}
\SetAlCapFnt{\small}
\SetAlCapNameFnt{\small}
\SetAlCapHSkip{0pt}
\IncMargin{-\parindent}

\providecommand{\keywords}[1]{\textbf{{Keywords.}} #1}

\begin{document}




\author{Hassan Jameel Asghar$^{1, 2}$, Ming Ding$^1$, Thierry Rakotoarivelo$^1$, Sirine Mrabet$^1$,\\
Mohamed Ali Kaafar$^{1, 2}$\\\\
\small $^1$Data61, CSIRO, Sydney, Australia\\
\small \texttt{\{hassan.asghar, ming.ding, thierry.rakotoarivelo, sirine.mrabet, dali.kaafar\}@data61.csiro.au}\\
\small $^2$Macquarie University, Sydney, Australia\\
\small \texttt{\{hassan.asghar, dali.kaafar\}@mq.edu.au}
}

\title{Differentially Private Release of High-Dimensional Datasets using the Gaussian Copula}

\maketitle

%
%
%
%

\begin{abstract}
We propose a generic mechanism to efficiently release differentially private synthetic versions of high-dimensional datasets with high utility. The core technique in our mechanism is the use of copulas, which are functions representing dependencies among random variables with a multivariate distribution. Specifically, we use the Gaussian copula to define dependencies of attributes in the input dataset, whose rows are modelled as samples from an unknown multivariate distribution, and then sample synthetic records through this copula. Despite the inherently numerical nature of Gaussian correlations we construct a method that is applicable to both numerical and categorical attributes alike. Our mechanism is efficient in that it only takes time proportional to the square of the number of attributes in the dataset. We propose a differentially private way of constructing the Gaussian copula without compromising computational efficiency. Through experiments on three real-world datasets, we show that we can obtain highly accurate answers to the set of all one-way marginal, and two-and three-way positive conjunction queries, with 99\% of the query answers having absolute (fractional) error rates between 0.01 to 3\%. Furthermore, for a majority of two-way and three-way queries, we outperform independent noise addition through the well-known Laplace mechanism. In terms of computational time we demonstrate that our mechanism can output synthetic datasets in around 6 minutes 47 seconds on average with an input dataset of about 200 binary attributes and more than 32,000 rows, and about 2 hours 30 mins to execute a much larger dataset of about 700 binary attributes and more than 5 million rows. To further demonstrate scalability, we ran the mechanism on larger (artificial) datasets with 1,000 and 2,000 binary attributes (and 5 million rows) obtaining synthetic outputs in approximately 6 and 19 hours, respectively. These are highly feasible times for synthetic datasets, which are one-off releases. 
\end{abstract}


\section{Introduction}
\label{sec:intro}
There is an ever increasing demand to release and share datasets owing to its potential benefits over controlled access. For instance, once data is released, data custodian(s) need not worry about access controls and continual support.
From a usability perspective, data release is more convenient for users (expert analysts and novices alike) as compared to access through a restricted interface. Despite its appeal, sharing datasets, especially when they contain sensitive information about individuals, has privacy implications which have been well documented. Current practice, therefore, suggests privacy-preserving release of datasets. Ad hoc techniques such as de-identification, which mainly rely on properties of datasets and assumptions on what background information is available, have failed to guarantee privacy~\cite{netflix, ohm-broken}. Part of the reason for the failure is the lack of a robust definition of privacy underpinning these techniques.

This gave rise to the definition of differential privacy~\cite{calib-noise}. Differential privacy ties the privacy property to the process or algorithm (instead of the dataset) and, informally, requires that any output of the algorithm be almost equally likely even if any individual's data is added or removed from the dataset. A series of algorithms have since been proposed to release differentially private datasets, often termed as synthetic datasets. Simultaneously, there are results indicating that producing synthetic datasets which accurately answer a large number of queries is computationally infeasible, i.e., taking time exponential in the number of attributes (dimension of the dataset)~\cite{pcp-hard}. This is a serious roadblock as real-world datasets are often high-dimensional. However, infeasibility results in~\cite{pcp-hard,ullman-n2} are generic, targeting \textit{provable} utility for \textit{any} input data distribution. It may well be the case that a large number of real-world datasets follow constrained distributions which would make it computationally easier to output differentially private synthetic datasets that can accurately answer a larger number of queries. Several recent works indicate the plausibility of this approach claiming good utility in practice~\cite{privbayes, dp-copula}. Ours is a continuation of this line of work.

In this paper, we present a \emph{generic} mechanism to \emph{efficiently} generate differentially private synthetic versions of high dimensional datasets with \emph{good utility}. By efficiency, we mean that our mechanism can output a synthetic dataset in time $O(m^2 n)$, where $m$ is the total number of attributes in the dataset and $n$ the total number of rows. Recall that the impossibility results~\cite{pcp-hard} suggest that algorithms for accurately answering a large number of queries are (roughly) expected to run in time $\text{poly}(2^m, n)$. Thus, our method is scalable for high dimensional datasets, i.e., having a large $m$. In terms of utility, our generated synthetic dataset is designed to give well approximated answers to all one-way marginal and two-way positive conjunction queries. One-way marginal queries return the number of rows in the dataset exhibiting a given value $x$ or its negation $\overline{x}$ (not $x$) under any attribute $X$.  Similarly, two-way positive conjunction queries return the number of rows that exhibit any pair of values $(x, y)$ under an attribute pair $(X, Y)$. This forms a subset of all two-way margins; the full set also includes negations of values, e.g., rows that satisfy $(x, \overline{y})$. While this may seem like a small subset of queries, there are results showing that even algorithms that generate synthetic datasets to (accurately) answer all two-way marginal queries are expected to be computationally inefficient~\cite{pcp-hard, ullman-n2}. Furthermore, we show that our mechanism provides good utility for other queries as well by evaluating the answers to 3-way positive conjunction queries on the synthetic output. 

The key to our method is the use of copulas~\cite{nelsen} to generate synthetic datasets. Informally, a copula is a function that maps the marginal distributions to the joint distribution of a multivariate distribution, thus defining dependence among random variables. Modelling the rows of the input dataset as samples of the (unknown) multivariate distribution, we can use copulas to define dependence among attributes (modelled as univariate random variables), and finally use it to sample rows from the target distribution and generate synthetic datasets. Specifically, we use the Gaussian copula~\cite[p. 23]{nelsen}, which defines the dependence between one-way margins through a covariance matrix. The underlying assumption is that the relationship between different attributes of the input dataset is completely characterised by pairwise covariances. While this may not be true in practice, it still preserves the correlation of highly correlated attributes in the synthetic output. Our main reason for using the Gaussian copula is its efficiency, as its run-time is proportional to square of the number of attributes. We remark that our use of the Gaussian copula \emph{does not} require the data attributes to follow a Gaussian distribution. Only their dependencies are assumed to be captured by the Gaussian copula.

Importantly, as claimed, our method is generic. This is important since an input dataset is expected to be a mixture of numerical (ordinal) and categorical (non-ordinal) attributes meaning that we cannot use a unified measure of correlation between attributes. A common technique to work around this is to create an artificial order on the categorical attributes~\cite{tax-tree}. However, as we show later, the resulting correlations are inherently artificial and exhibit drastically different results if a new arbitrary order is induced. Our approach is to convert the dataset into its binary format (see Section~\ref{sub:dummy})
and then use a single correlation measure, Pearson's product-moment correlation, for all binary attributes. This eliminates the need for creating an artificial order on the categorical attributes. Our method is thus more generic than a similar method proposed in~\cite{dp-copula}, which only handles attributes with large (discrete) domains and induces artificial order on categorical attributes.\footnote{See Section~\ref{sec:rw} for a further discussion on the differences between the two works.}

We experimentally evaluate our method on three real-world datasets: the publicly available Adult dataset \cite{adult} containing US census information, a subset of the social security payments dataset provided to us by the Department of Social Services (DSS), a department of the Government of Australia,\footnote{\url{https://www.dss.gov.au/}.}, and a hospital ratings dataset extracted from a national patient survey in the United States, which we call the Hospital dataset. The Adult dataset consists of 14 attributes (194 in binary) and more than 32,000 rows, the {\dss} dataset has 27 attributes (674 in binary) and more than 5,000,000 rows, and the Hospital dataset contains 9 attributes (1,201 in binary) and 10,000 rows. Generation of synthetic datasets took around 6 minutes 47 seconds (on average) for the Adult dataset, around two and a half hours on average for the {\dss} dataset, and 46 minutes on average for the Hospital dataset. To further check the scalability of our method, we ran it on two artificial datasets each with more than 5 million rows and 1,000 and 2,000 binary attributes, resulting in run-times of approximately 6 and 19 hours, respectively. Since synthetic data release is a one-off endeavour, these times are highly feasible. In terms of utility, we show that 99\% of the queries have an absolute error of less than 500 on Adult, less than 1,000 on {\dss} and around 300 on the Hospital dataset, where absolute error is defined as the absolute difference between true and differentially private answers computed from the synthetic dataset. In terms of two-way positive conjunction queries, we again see that 99\% of the queries have an absolute error of around 500 for Adult, less than 500 for {\dss} and only around 20 for the Hospital dataset. Furthermore, for most of the two-way queries we considerably outperform the Laplace mechanism~\cite{calib-noise} of adding independent Laplace noise to each of the queries. Note that a further advantage of our method is that unlike the Laplace mechanism, we generate a synthetic dataset. We further expand our utility analysis to include three-way queries and show that our method again calibrates noise considerably better than the aforementioned Laplace mechanism with 99\% of the queries having absolute error of around 400 for both Adult and {\dss}, and only around 50 for the Hospital dataset. Our utility analysis is thorough; we factor in the possibility that real-world datasets may have significantly high number of uncorrelated attributes which implies that a synthetic data generation algorithm might produce accurate answers by chance (for the case of two-way or higher order marginals). Thus we separate results for highly correlated and uncorrelated attributes to show the accuracy of our mechanism. Perhaps one drawback of our work is the lack of a theoretical accuracy bound; but this is in line with many recent works~\cite{privbayes, psd, dp-copula} which, like us, promise utility in practice.

The rest of the paper is organized as follows. We give a brief background on differential privacy and copulas in Section~\ref{sec:bg}. We describe our mechanism in Section~\ref{sec:method}. Section~\ref{sec:exp} contains our experimental utility and performance analysis. We present related work in Section~\ref{sec:rw}, and conclude in Section~\ref{sec:conclude}. 
\section{Background concepts}
\label{sec:bg}

\subsection{Notations}

We denote the original dataset by $D$, which is modelled as a multiset of $n$ rows from the domain $\domain = A_1 \times \cdots \times A_m$,
where each $A_j$ represents an attribute, for a total of $m$ attributes. We assume the number of rows $n$ is \emph{publicly known}. We denote the set of all $n$-row databases as $\domain^n$. Thus, $D \in \domain^n$.
The $i$th row of $D$ is denoted as $(X^{(i)}_1, X^{(i)}_2, \ldots, X^{(i)}_m)$,
where $X^{(i)}_j \in A_j$.
In the sequel,
where discussing a generic row,
we will drop the superscript for brevity. The notation $\overline{\mathbb{R}}$ represents the real number line $[-\infty, \infty]$, and $\mathbb{I}$ denotes the interval of real numbers $[0, 1]$. The indicator function $I\{P\}$ evaluates to $1$ if the predicate $P$ is true, and $0$ otherwise.

\subsection{Dummy Coding}
\label{sub:dummy}
A key feature of our method is to convert the original dataset into its binary format through dummy coding, where each attribute value of $A_j \in D$ is assigned a new binary variable. For instance, consider the ``country'' attribute having attribute values \texttt{USA}, \texttt{France} and \texttt{Australia}. Converted to binary we will have three columns (binary attributes) one for each of the three values. For any row, a value of 1 corresponding to any of these binary columns indicates that the individual is from that particular country. Thus, we have a total of $d = \sum_{i = 1}^m | A_i |$ binary attributes in the binary version $D_{\mathsf{B}}$ of the dataset $D$. In case of continuous valued attributes, this is done by first binning the values into discrete bins. Thus for a continuous valued attribute $A$, $|A|$ indicates the number of bins. The $i$th row in $D_{\mathsf{B}}$ is denoted as $(X^{(i)}_1, X^{(i)}_2, \ldots, X^{(i)}_d)$. Note that while the mapping from $D$ to $D_{\mathsf{B}}$ is unique, the converse is not always true.\footnote{For instance, in $D_{\mathsf{B}}$, we may have two binary attributes set to 1, which correspond to two different values of the same attribute in the original dataset $D$, e.g., \texttt{USA} and \texttt{France}.}  Note further that this way of representing datasets is also known as histogram representation~\cite[\S 2.3]{dp-book}.


\subsection{Overview of Differential Privacy}
Two databases $D, D' \in \domain^n$ are called neighboring databases if they differ in only one row.
\begin{definition}[Differential privacy~\cite{calib-noise, dp-book}]
A randomized algorithm (mechanism) $\mathcal{M}: \domain^n \rightarrow R$ is $(\epsilon, \delta)$-differentially private if for every $S \subseteq R$, and for all neighbouring databases $D, D' \in \domain^n$, the following holds \[ \mathbb{P} (\mathcal{M}(D) \in S ) \le e^{\epsilon} \mathbb{P} (\mathcal{M}(D') \in S) + \delta. \] If $\delta = 0$, then $\mathcal{M}$ is $\epsilon$-differentially private.
\end{definition}
The mechanism $\mathcal{M}$ might also take some auxiliary inputs such as a query or a set of queries (to be defined shortly). The parameter $\delta$ is required to be a negligible function of $n$ \cite[\S 2.3, p. 18]{dp-book},\cite[\S 1.6, p. 9]{salil-tut}.\footnote{A function $f$ in $n$ is negligible, if for all $c \in \mathbb{N}$, there exists an $n_0 \in \mathbb{N}$ such that for all $n \ge n_0$, it holds that $f(n) < n^{-c}$.} The parameter $\epsilon$ on the other hand should be small but not arbitrarily small. We may think of $\epsilon \le 0.01$, $\epsilon \le 0.1$ or $\epsilon \le 1$~\cite[\S 1]{perplexed}, \cite[\S 3.5.2, p. 52]{dp-book}.
An important property of differential privacy is that it composes~\cite{dp-book}.
\begin{theorem}[Basic composition]
\label{the:basic-comp}
If $\mathcal{M}_1, \ldots, \mathcal{M}_k$ are each $(\epsilon, \delta)$-differentially private then $\mathcal{M} = ( \mathcal{M}_1, \ldots, \mathcal{M}_k)$ is $(k\epsilon, k\delta)$-differentially private. \qed
\end{theorem}
The above is sometimes referred to as (basic) \emph{sequential} composition as opposed to parallel composition defined next.
\begin{theorem}[Parallel composition~\cite{mcsherry-par-comp}]
\label{the:par-comp}
Let $\mathcal{M}_i$ each provide $(\epsilon, \delta)$-differential privacy. Let $\domain_i$ be arbitrary disjoint subsets of the domain $\domain$. The sequence of $\mathcal{M}_i(D \cap \domain_i)$ provides $(\epsilon, \delta)$-differential privacy, where $D \in \domain^n$.\qed
\end{theorem}
A more advanced form of composition, which we shall use in this paper, only depletes the ``privacy budget'' by a factor of $\approx \sqrt{k}$. See~\cite[\S 3.5.2]{dp-book} for a precise definition of adaptive composition.
\begin{theorem}[(Advanced) adaptive composition~\cite{boosting, dp-book, dual-query}]
\label{the:comp}
Let $\mathcal{M}_1, \ldots, \mathcal{M}_k$ be a sequence of mechanisms where each can take as input the output of a previous mechanism, and let each be $\epsilon'$-differentially private. Then $\mathcal{M}(D) =  (\mathcal{M}_1(D), \ldots, \mathcal{M}_k(D)$ is $\epsilon$-differentially private for $\epsilon = k\epsilon'$, and $(\epsilon, \delta)$-differentially private for
\[
\epsilon = \sqrt{2k \ln{\frac{1}{\delta}}} \epsilon' + k \epsilon' (e^{\epsilon'} - 1),
\]
for any $\delta \in (0, 1)$. Furthermore, if the mechanisms $\mathcal{M}_i$ are each $(\epsilon', \delta')$ differentially private, then $M$ is $(\epsilon, k\delta' + \delta)$ differentially private with the same $\epsilon$ as above.\qed
\end{theorem}

Another key feature of differential privacy is that it is immune to post-processing.
\begin{theorem}[Post-processing~\cite{dp-book}]
\label{the:post-proc}
If $\mathcal{M} : \domain^n \rightarrow R$ is $(\epsilon, \delta)$-differentially private and $f : R \rightarrow R'$ is any randomized function, then $f \circ \mathcal{M} : \domain^n \rightarrow R'$ is $(\epsilon, \delta)$-differentially private.\qed
\end{theorem}

A query is defined as a function $q: \domain^n \rightarrow \mathbb{R}$.
\begin{definition}[Counting queries and point functions]
\label{def:pt}
A \emph{counting} query is specified by a predicate $q: \domain \rightarrow \{0, 1\}$ and extended to datasets $D \in \domain^n$ by summing up the predicate on all $n$ rows of the dataset as
\[
q(D) = \sum_{x \in D} q(x).
\]
A point function~\cite{salil-tut} is the sum of the predicate $q_y: \domain \rightarrow \{0, 1\}$, which evaluates to $1$ if the row is equal to the point $y \in \domain$ and $0$ otherwise, over the dataset $D$. Note that computing all point functions, i.e., answering the query $q_y(D)$ for all $y \in \domain$, amounts to computing the histogram of the dataset $D$.
\qed
\end{definition}

\begin{definition}[Global sensitivity~\cite{dp-book, salil-tut}]
The global sensitivity of a counting query $q: \domain^n \rightarrow \mathbb{N}$ is
\[
\Delta q = \max_{\substack{D, D' \in \domain^n \\ D \sim D'}} | q(D) - q(D') |.
\]
\qed
\end{definition}

\begin{definition}[$(\alpha, \beta)$-utility]
The mechanism $\mathcal{M}$ is said to be $(\alpha, \beta)$-useful for the query class $Q$ if for any $q \in Q$,
\[
\mathbb{P} \left[ \mathsf{util} (\mathcal{M}(q, D), q(D)) \le \alpha \right] \ge 1- \beta,
\]
where the probability is over the coin tosses of $\mathcal{M}$, and $\mathsf{util}$ is a given metric for utility.
\end{definition}

The Laplace mechanism is employed as a building block in our algorithm to generate the synthetic dataset.
\begin{definition}[Laplace mechanism~\cite{calib-noise}]
\label{def:laplace}
The Laplace distribution with mean $0$ and scale $b$ has the probability density function
\[
\text{Lap}(x \mid b) = \frac{1}{2b}e^{-\frac{|x|}{b}}.
\]
We shall remove the argument $x$, and simply denote the above by $\text{Lap}(b)$. Let $q: \domain^n \rightarrow \mathbb{R}$ be a query. The mechanism
\[
\mathcal{M}_{\text{Lap}}(q, D, \epsilon) = q(D) + Y
\]
where $Y$ is drawn from the distribution $\text{Lap}\left(\Delta q/ \epsilon \right)$ is known as the Laplace mechanism.
The Laplace mechanism is $\epsilon$-differentially private~\cite[\S 3.3]{dp-book}. Furthermore, with probability at least $1 - \beta$~\cite[\S 3.3]{dp-book} \[ \max_{q \in Q} | q(D) - \mathcal{M}_{\text{Lap}}(q, D, \epsilon) | \le \frac{\Delta q}{\epsilon} \ln \left(\frac{1}{\beta} \right) \doteq \alpha, \] where $\beta \in (0, 1]$. If $q$ is a counting query, then $\Delta q = 1$, and the error $\alpha$ is \[ \alpha = \frac{1}{\epsilon} \ln \left(\frac{1}{\beta} \right) \] with probability at least $1 - \beta$. \qed
\end{definition}


\subsection{Overview of Copulas}
For illustration, we assume the bivariate case, i.e., we have only two variables (attributes) $X_1$ and $X_2$. Let $F_1$ and $F_2$ be their margins, i.e., cumulative distribution functions (CDFs), and let $H$ be their joint distribution function. Then, according to Sklar's theorem~\cite{sklar, nelsen}, there exists a copula $C$ such that for all $X_1, X_2 \in \overline{\mathbb{R}}$,
\[
H(X_1, X_2) = C(F_1(X_1), F_2(X_2)),
\]
which is unique if $F_1$ and $F_2$ are continuous; otherwise it is uniquely determined on $\text{Ran}(F_1) \times \text{Ran}(F_2)$, where $\text{Ran}(\cdot)$ denotes range. In other words, there is a function that maps the joint distribution function to each pair of values of its margins. This function is called a copula. Since our treatment is on binary variables $X \in \{0, 1\}$, the corresponding margins are defined over $X \in \overline{\mathbb{R}}$ as
\[
F(X) =
\begin{cases}
0, & X \in [-\infty, 0) \\
a, & X \in [0, 1) \\
1, & X \in [1, \infty]
\end{cases}
\]
where, $a \in \mathbb{I}$. The above satisfies the definition of a distribution function~\cite[\S 2.3]{nelsen}. This allows us to define the \emph{quasi-inverse} of $F$, denoted $F^{-1}$, (tailored to our case) as
\[
F^{-1}(t) =
\begin{cases}
0, & t \in [0, a] \\
1, & t \in (a, 1]
\end{cases}
\]
Now, using the quasi-inverses, we see that
\[
C(u, v) = H(F_1^{-1}(u), F_2^{-1}(v)),
\]
where $u, v \in \mathbb{I}$. If an analytical form of the joint distribution function is known, the copula can be constructed from the expressions of $F_1^{-1}$ and $F_2^{-1}$ (provided they exist). This then allows us to generate random samples of $X_1$ and $X_2$ by first sampling a uniform $u \in \mathbb{I}$, and then extracting $v$ from the conditional distribution. Using the inverses we can extract the pair $X_1, X_2$. See~\cite[\S 2.9]{nelsen} for more details. However, in our case, and in general for real-world datasets, we seldom have an analytical expression for $H$, which could allow us to construct $C$. There are two approaches to circumvent this, which rely on obtaining an empirical estimate of $H$ from $D_{\mathsf{B}}$.

\begin{itemize}
  \item The first approach relies on constructing a discrete equivalent of an empirical copula~\cite[\S 5.6, p. 219]{nelsen}, using the empirical $H$ obtained from the dataset $D_\mathsf{B}$. However, doing this in a differentially private manner requires computing answers to the set of all point functions (Definition~\ref{def:pt}) of the original dataset~\cite{salil-tut}, which amounts to finding the histogram of $D_\mathsf{B}$. Unfortunately, for high dimensional datasets, existing \emph{efficient} approaches of differentially private histogram release~\cite{salil-tut} would release a private dataset that discards most rows of the original dataset. This is due to the fact that high dimensional datasets are expected to have high number of rows with low multiplicity.
  \item The other approach, and indeed the one taken by us, is using some existing copula and adjusting its parameters according to the dataset $D_\mathsf{B}$. For this paper we choose the Gaussian copula, i.e., the copula
\[
C(u, v) = \mathbf{\Phi}_{r}( \Phi^{-1}(u), \Phi^{-1}(v)),
\]
where $\Phi$ is the standard univariate normal distribution, $r$ is the Pearson correlation coefficient and $\mathbf{\Phi}_r$ is the standard bivariate normal distribution with correlation coefficient $r$. We can then replace the $\Phi$'s with the given marginals $F$ and $G$ resulting in a distribution which is not standard bivariate, if $F$ and $G$ are not standard normal themselves. The underlying assumption in this approach is that the given distribution $H$ is completely characterised by the margins and the correlation. This in general may not be true of all distributions $H$. However, in practice, this can potentially provide good estimates of one way margins and a subset of the two way margins as discussed earlier. Another advantage of this approach is computational efficiency: the run time is polynomial in $m$ (the number of attributes) and $n$ (the number or rows).
\end{itemize}

While our introduction to copulas has focused on two variables, this can be extended to multiple variables~\cite{sklar}.

\section{Proposed Mechanism}
\label{sec:method}
Our method has three main steps as shown in Figure~\ref{fig:flowchart}: data pre-processing, differentially private statistics, and copula construction. Among these, only the second step involves privacy treatment. The last step preserves privacy due to the post-processing property of differential privacy (Theorem~\ref{the:post-proc}). The first step, if not undertaken carefully, may result in privacy violations, as we shall discuss shortly. We shall elaborate each step in the following.

\begin{figure*}[!htb]
\centering
\includegraphics[width=0.9\textwidth]{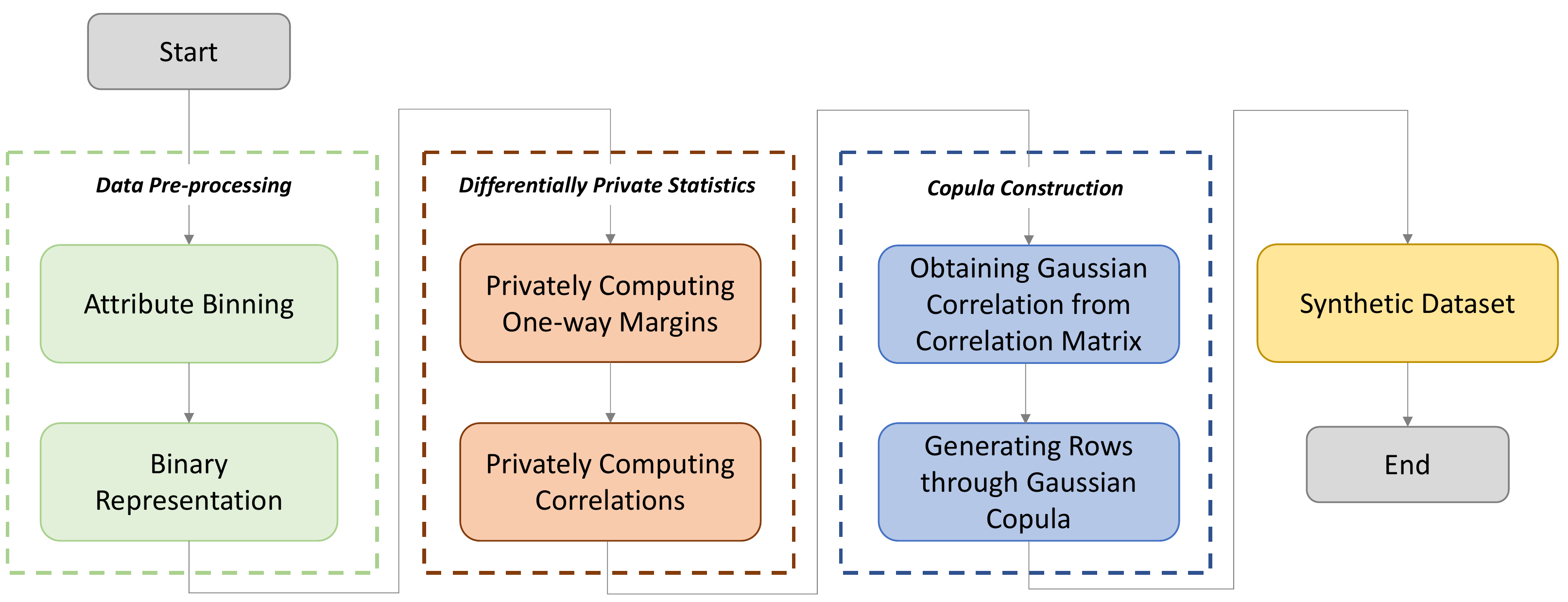}
\caption{Flowchart of our method.}
\label{fig:flowchart}
\end{figure*}


\subsection{Data Pre-processing}
\label{s:preproc}

\subsubsection{Binning of Continuous and Ordinal Data}
Since we will convert the original data into its binary format we need to ensure that the resulting expansion, i.e., $d = \sum_{i = 1}^m | A_i |$, does not become prohibitively large. This will in general be the case with continuous attributes or discrete valued attributes with large domains. To overcome this, we bin these attributes into discrete intervals. However, care must be taken to ensure that the resulting binning is done independent of the input dataset $D$ to avoid privacy violations. For instance, data-specific binning could introduce a bin due to an outlier which would reveal the presence and absence of that particular individual, violating differential privacy. We, therefore, assume that this binning is done only through what is publicly known about each attribute. For instance, when binning the year of birth into bins of 10 years, we assume that most individuals in the dataset have an age less than 120 (without checking this in the dataset). This of course depends on the nature of the attributes. While binning can be done in a differentially private manner, we use a simpler approach as this is not the main theme of our paper. We also note that data binning as a pre-processing step is almost always adopted in related work, e.g.,~\cite{dual-query}.

\subsubsection{Binary Representation}
In order to use the Gaussian copula, we need to determine correlations between attributes via some fixed correlation \emph{coefficient}; which is a numeric measure. We use the Pearson product moment correlation coefficient for this purpose. This means that in order to measure correlations, either the attribute values in the dataset need to be real valued or have to mapped to a real number. These mapped values then follow their order in the set of real numbers. If an attribute is ordinal, e.g., age or salary, a natural order exists. However, categorical attributes do not have a natural order, e.g., gender. One way to solve this conundrum is to induce an artificial order among the values of the categorical attribute, e.g., through a hierarchical tree~\cite{tax-tree}. However, this ordering is inherently artificial and any change in order results in a markedly different correlation. An illustration of this fact is shown in Appendix~\ref{app:cor-art-ord} (we also briefly discuss this in Section~\ref{sec:rw}).\footnote{There are other correlation measures that can be used for categorical attributes, e.g., Cram{\'e}r's V~\cite{cramer}. However, we need a unified measure that is applicable to all attributes alike.} Our approach instead is to work on the binary version of the dataset obtained via dummy coding, as explained in Section~\ref{sub:dummy}. Pairwise correlations now amount to finding Pearson correlation between pairs of binary variables. This way of interpreting data makes no distinction between ordinal and categorical variables by not assigning any order to the latter.

\subsection{Differentially Private Statistics}

\subsubsection{Privately Computing One-Way Margins}
\label{subsub:priv-comp-marg}

The first ingredient in our method is the one-way margins, i.e., marginal CDFs, of the attributes $X_i$, for $i \in \left\{ 1,2,\ldots,d\right\}$. We denote these margins by $\hat{F}_{i}\left(x\right)$,
where $x \in \{0, 1\}$. Since each $X_i$ is binary, these margins can be calculated as
\begin{align*}
\hat{F}_{i}\left(0\right) = \frac{\sum_{k=1}^{n} I \{ X_{i}^{(k)} = 0 \}} {n}, \quad \hat{F}_{i}\left(1\right) = 1
\end{align*}

To make $\hat{F}_{i}\left(x\right)$ differentially private, we add Laplace noise to the counts $\hat{n}_{i}^{0} = \sum_{k=1}^{n} I \{ X_{i}^{(k)} = 0 \}$ and $\hat{n}_{i}^{1} \sum_{k=1}^{n} I \{ X_{i}^{(k)} = 1 \}$ to obtain $\tilde{F}_{i}\left(x\right)$, which is summarized in Algorithm~\ref{alg:algo_DP_emp_CDF}. If the differentially private sum is negative, we fix it to 0. Note that this utilizes the post-processing property (cf. Theorem~\ref{the:post-proc}) and hence maintains privacy. The reason we add noise to both $\hat{n}_{i}^{0}$ and $\hat{n}_{i}^{1}$ instead of just adding noise to $\hat{n}_{i}^{0}$ is to avoid the noisy $\hat{n}_{i}^{0}$ exceeding $n$.

\begin{algorithm}
\caption{\label{alg:algo_DP_emp_CDF} Obtaining Differentially Private Marginals $\tilde{F}_{i}\left(x\right)$}

\begin{enumerate}
\item For the $i$th attribute, count the numbers of events $X_{i}=0$
and $X_{i}=1$ as
\begin{equation*}
\hat{n}_{i}^{0}=\sum_{k=1}^{n} I \left\{ X_{i}^{(k)} =0 \right\},  \quad \hat{n}_{i}^{1}=\sum_{k=1}^{n} I \left\{ X_{i}^{(k)} =1 \right\}
\end{equation*}
\item Add Laplace noise to $\hat{n}_{i}^{0}$ and $\hat{n}_{i}^{1}$, and obtain the noisy counts as
\begin{equation*}
\tilde{n}_{i}^{0}=\hat{n}_{i}^{0}+\text{Lap}\left(\frac{2}{\epsilon'_i}\right), \quad \tilde{n}_{i}^{1}=\hat{n}_{i}^{1}+\text{Lap}\left(\frac{2}{\epsilon'_i}\right)
\end{equation*}
where $\epsilon'_i$ is the privacy budget associated with computing the $i$th margin.
\item Obtain $\tilde{F}_{i}\left(x\right)$ as
\begin{equation*}
\tilde{F}_{i}\left(0\right)=\frac{\tilde{n}_{i}^{0}}{\tilde{n}_{i}^{0} + \tilde{n}_{i}^{1}}, \quad \tilde{F}_{i}\left(1\right)=1
\end{equation*}
\end{enumerate}
\end{algorithm}

Algorithm~\ref{alg:algo_DP_emp_CDF} is $(\epsilon'_i, 0)$-differentially private due to differential privacy of the Laplace mechanism and the fact that the algorithm is essentially computing the histogram associated with attribute $X_i$~\cite[\S 3.3, p. 33]{dp-book}.
Importantly, the privacy budget $\epsilon'_i$ is impacted only by the number of attributes $m$ in the original dataset $D$ and not by the number of binary attributes $d$ in the binary version $D_{\mathsf{B}}$.
This is shown by the following lemma.

\begin{lemma}
\label{lem:perturb_CDF}
Fix an attribute $A$ in $D$, and let $X_1, \ldots, X_{|A|}$ denote the binary attributes constructed from $A$. Then if the computation of each marginal $X_j$, $j \in [|A|]$, is $(\epsilon'_i, 0)$-differentially private, the computation of the marginal $A$ is $(\epsilon'_i, 0)$-differentially private.
\end{lemma}
\begin{proof}
See Appendix~\ref{app:lem:perturb_CDF}.
\end{proof}

\subsubsection{Privately Computing Correlations}
The other requirement of our method is the computation of pairwise correlations given by
\begin{equation}
\label{eq:emp-cor}
\hat{r}_{i,j}=\frac{\mathbb{E}\left\{ X_{i}X_{j}\right\} -\mu_{i}\mu_{j}}{\sqrt{\textrm{var}\left(X_{i}\right)\textrm{var}\left(X_{j}\right)}}, \quad  i, j \in \{ 1,2,\ldots,d\}.
\end{equation}
To obtain the differentially private version of $\hat{r}_{i,j}$, denoted $\tilde{r}_{i,j}$, one way is to compute $\hat{r}_{i,j}$ directly from $D_{\mathsf{B}}$ and then add Laplace noise scaled to the sensitivity of $\hat{r}$. However, as we show in Appendix~\ref{app:cor-high-sen}, the empirical correlation coefficient from binary attributes has high global sensitivity, which would result in highly noisy $\tilde{r}_{i,j}$. The other approach is to compute each of the terms in Eq.~\ref{eq:emp-cor} in a differentially private manner and then obtain $\tilde{r}_{i,j}$. Notice that for binary attribute $X_i$, its mean $\mu_i$ is given by $\hat{n}_{i}^{1}/n$. This can be obtained differentially privately as
\[
\tilde{\mu}_i = 1 - \tilde{F}_{i}\left(0\right)=1 - \frac{\tilde{n}_{i}^{0}}{\tilde{n}_{i}^{0} + \tilde{n}_{i}^{1}},
\]
which we have already computed. Likewise, the variance $\hat{\text{var}}(X_i)$ is given by $\hat{\mu}_i (1 - \hat{\mu}_i)$, whose differentially private analogue, i.e.,  $\tilde{\text{var}}(X_i)$, can again be obtained from the above. Thus, the only new computation is the computation of $\mathbb{E}\left\{ X_{i}X_{j}\right\}$, which for binary attributes is equivalent to computing $\frac{1}{n} \sum_{k = 1}^{n} I \{ (X_i^{(k)} = 1) \wedge (X_j^{(k)} = 1) \}$. Algorithm~\ref{alg:algo_DP_emp_joint_pair_CDF} computes this privately.

\begin{algorithm}
\caption{\label{alg:algo_DP_emp_joint_pair_CDF} Obtaining Differentially Private Two-Way Positive Conjunctions $\tilde{\mathbb{E}}\left\{ X_{i}X_{j}\right\}$}

\begin{enumerate}
\item For the $i$th and $j$th attributes  ($1 \le i < j \le d$), and for $a, b \in \{0, 1\}$, count the number of events $(X_{i}, X_{j}) = (a, b)$ as
\begin{equation}
\hat{n}_{i,j}^{ab}=\sum_{k=1}^{n} I \left\{ (X_{i}^{(k)} = a) \cap (X_{j}^{(k)} = b) \right\}
\end{equation}
\item Add Laplace noise onto $\hat{n}_{i,j}^{ab}$ for $a, b \in \{0, 1\}$ and obtain the noisy count as
\begin{equation}
\tilde{n}_{i,j}^{ab}=\hat{n}_{i,j}^{ab}+\text{Lap}\left(\frac{2}{\epsilon''_{i,j}}\right),
\end{equation}
where $\epsilon''_{i,j}$ is the privacy budget associated with computing the $(i, j)$th two-way marginal.
\item Obtain $\tilde{\mathbb{E}}\left\{ X_{i}X_{j}\right\}$ as
\begin{equation}
\tilde{\mathbb{E}}\left\{ X_{i}X_{j}\right\} = \frac{\tilde{n}_{i,j}^{ab}}{\sum_{a, b \in \{0, 1\} }\tilde{n}_{i,j}^{ab}}.
\end{equation}
\end{enumerate}
\end{algorithm}

Algorithm~\ref{alg:algo_DP_emp_joint_pair_CDF} is $(\epsilon''_{i,j}, 0)$-differentially private due to the differential privacy of the Laplace mechanism and the fact that the algorithm is essentially computing the histogram associated with attribute pairs $(X_i, X_j)$~\cite[\S 3.3, p. 33]{dp-book}.
Once again, the privacy budget $\epsilon''_{i, j}$ is impacted only by the number of pairs of attributes $\binom{m}{2}$ in the original dataset $D$ and not by the number of pairs of binary attributes $\binom{d}{2}$ in the binary version $D_{\mathsf{B}}$.
This is presented in the following lemma,
whose proof is similar to that of Lemma~\ref{lem:perturb_CDF}, and hence omitted for brevity.

\begin{lemma}
\label{lem:perturb_2CDF}
Fix two attributes $A_i$ and $A_j$, $i \ne j$, in $D$. Let $X_{i, 1}, \ldots, X_{i, |A_i|}$ and  $X_{j, 1}, \ldots, X_{j, |A_j|}$ denote the binary attributes constructed from $A_i$ and $A_j$, respectively. Then, if the computation of each of the two-way marginals $(X_{i, k}, X_{j, k'})$, $k \in [|A_i|], k' \in [|A_j|]$, is $(\epsilon''_{i,j}, 0)$-differentially private, the computation of all all two-way marginals of $A_i$ and $A_j$ is $(\epsilon''_{i,j}, 0)$-differentially private.\qed
\end{lemma}

The differentially private correlation coefficients $\tilde{r}_{i,j}$ thus obtained can be readily used to construct the differentially private correlation matrix
\begin{equation}
\tilde{\mathbf{R}}=\left[\begin{array}{cccc}
\tilde{r}_{1,1} & \tilde{r}_{1,2} & \cdots & \tilde{r}_{1,d}\\
\tilde{r}_{2,1} & \tilde{r}_{2,2} & \cdots & \tilde{r}_{2,d}\\
\vdots & \vdots & \ddots & \vdots\\
\tilde{r}_{d,1} & \tilde{r}_{d,2} & \cdots & \tilde{r}_{d,d}
\end{array}\right].\label{eq:DP_emp_corr_matrix}
\end{equation}

Notice that the above algorithm computes the correlations in time $O(d^2 n)$. This can be prohibitive if $d$ is large, i.e., if each of the $m$ (original) attributes have large domains. An alternative algorithm to compute the two-way positive conjunctions that takes time only $O(m^2 n)$ is shown in Appendix~\ref{app:faster-two-way}.

\subsection{Copula Construction}
\label{sub:cop-construct}
For this section, we do not need access to the database any more. Hence,
any processing done preserves the previous privacy budget due to closure under post-processing (Theorem~\ref{the:post-proc}).

\subsubsection{Obtaining Gaussian Correlation from the Correlation Matrix}
\label{subsec:algo_orig2Gauss_mapping}
Our aim is to sample standard normal Gaussian variables $Y_{i}$'s corresponding to the attributes $\tilde{X}_{i}$'s\footnote{i.e., the synthetic versions of the $X_i$'s.} where the
correlations among $Y_{i}$'s, given by the correlation matrix $\mathbf{P}$, are \emph{mapped} to the correlations among the $X_{i}$'s, given by the (already computed) correlation matrix $\tilde{\mathbf{R}}$.
A sample from the Gaussian variable $Y_{i}$ is transformed
backed to $\tilde{X}_i$ as
\begin{equation}
\tilde{X}_{i}=\tilde{F}_{i}^{-1}\left(\Phi\left(Y_{i}\right)\right).\label{eq:RV_transformation}
\end{equation}
Obviously,
if the attributes $\tilde{X}_{i}$ are independent,
then $Y_{i}$ are also independent.
However,
in practice $\tilde{X}_{i}$'s are correlated,
which is characterized by $\tilde{\mathbf{R}}$ in our case.
Hence,
the question becomes: \emph{How to choose a correlation matrix $\mathbf{P}$ for $Y_{i}$'s,
so that $\tilde{X}_{i}$'s have the target correlation relationship defined by $\tilde{\mathbf{R}}$? }

From Eq.~\ref{eq:emp-cor} and its perturbation through $\tilde{r}_{i,j}$,
we can obtain
\begin{equation}
\mathbb{E}\left\{ \tilde{X}_{i}\tilde{X}_{j}\right\} =\tilde{r}_{i,j}\sqrt{\textrm{var}\left(\tilde{X}_{i}\right)\textrm{var}\left(\tilde{X}_{j}\right)}+\tilde{\mu}_{i}\tilde{\mu}_{j},\label{eq:cov_eq1}
\end{equation}
On the other hand,
from Eq.~\ref{eq:RV_transformation},
we can get
\begin{align}
\mathbb{E}\left\{ \tilde{X}_{i}\tilde{X}_{j}\right\}  & =  \mathbb{E}\left\{ \tilde{F}_{i}^{-1}\left(\Phi\left(Y_{i}\right)\right)\tilde{F}_{j}^{-1}\left(\Phi\left(Y_{j}\right)\right)\right\} \nonumber \\
 & =  \int_{-\infty}^{+\infty}\int_{-\infty}^{+\infty}\tilde{F}_{i}^{-1}\left(\Phi\left(y_{i}\right)\right)\tilde{F}_{j}^{-1}\left(\Phi\left(y_{j}\right)\right) \mathbf{\Phi}_{i,j}\left(y_{i},y_{j}\right)dy_{i}dy_{j},\label{eq:cov_eq2}
\end{align}
where $\mathbf{\Phi}_{i,j}\left(y_{i},y_{j}\right)$ denotes the standard bivariate probability density function (PDF) of the correlated standard normal random variables $Y_{i}$ and $Y_{j}$, given by
\begin{equation}
\mathbf{\Phi}_{i,j}\left(y_{i},y_{j}\right)=\frac{1}{2\pi\sqrt{1-\rho_{i,j}^{2}}}\exp\left(-\frac{y_{i}^{2}+y_{j}^{2}-2\rho_{i,j}y_{i}y_{j}}{2\left(1-\rho_{i,j}^{2}\right)}\right).\label{eq:joint_PDF_Gauss_simp}
\end{equation}
Here $\rho_{i, j}$ is the Pearson correlation coefficient, which for the standard normal variables $Y_i$ and $Y_j$ is given by $\mathbb{E}\left\{ Y_i Y_j \right\}$. Our task is to find the value of $\rho_{i, j}$ such that Eq.~\ref{eq:cov_eq1} and Eq.~\ref{eq:cov_eq2} are equal. In other words, Eqs.~\ref{eq:cov_eq1} and~\ref{eq:cov_eq2} define the relationship between $\tilde{r}_{i,j}$ and $\rho_{i, j}$. Notice that by construction, in general, we do \emph{not} have $\rho_{i,j}=\tilde{r}_{i,j}$. We can obtain $\rho_{i,j}$ by means of a standard bisection search (e.g., see~\cite[p. 148]{norta-book}). The two-fold integral in Eq.~\ref{eq:cov_eq2} with respect to $y_i$ and $y_j$, 
is evaluated numerically in the bisection search. In more detail, 
$dy_i$ and $dy_j$ are set to a small value of 0.01, 
and the lower and upper limits of the such integral are set to -10 and 10, respectively. 
Such lower and upper limits make the numerical results sufficiently accurate since $\mathbf{\Phi}_{i,j}\left(y_{i},y_{j}\right)$ is the standard bivariate PDF of two correlated normal random variables (and hence have negligibly small probability mass beyond the limits of integration). With each such $\rho_{i, j}$, we can construct the matrix $\mathbf{P}$ corresponding to $\tilde{\mathbf{R}}$.

\paragraph*{{Remark.}} Notice that the choice of $\rho_{i, j}$ ensures that the resulting sampled Gaussian variables have the property that when transformed back to $\tilde{X}_{i}$ and $\tilde{X}_{j}$, we get $\mathbb{E}\left\{ \tilde{X}_{i}\tilde{X}_{j}\right\} \approx \mathbb{E}\left\{{X}_{i}{X}_{j}\right\}$, where the latter is the input expectation. For binary attributes, recall that $\mathbb{E}\left\{ {X}_{i}{X}_{j}\right\} = \frac{1}{n} \sum_{k = 1}^{n} I \left\{ (X_i^{(k)} = 1) \cap (X_j^{(k)} = 1) \right\}$. Thus the method ensures that the ``11'''s are well approximated to the input distribution. However, the correlation coefficient, and the distribution of 01's, 10's and 00's might not be the same as the original dataset. This is evident from Eq.~\ref{eq:cov_eq1}. However, since the remaining quantities in Eq.~\ref{eq:cov_eq1} depend on the one-way margins, maintaining a good approximation to margins would imply that the distribution of 01's, 00's and 10's would also be well approximated. However, error in one or both of the margins would propagate to the error in the 01's, 10's and 00's. It is due to this reason that we target positive two-way conjunctions for utility.

\subsubsection{Generating Records Following a Multivariate Normal Distribution}
\label{subsec:algo_gen_multivariate_norm_distri}
Since each $Y_{i}$ follows the standard normal distribution and the correlation among $Y_{i}$'s is characterized by $\mathbf{P}$,
we can generate records by sampling from the resulting multivariate normal distribution. However, the matrix $\mathbf{P}$ obtained through this process should have two important properties for this method to work:
\begin{enumerate}
  \item $\mathbf{P}$ should be a correlation matrix, i.e., $\mathbf{P}$ should be a symmetric positive semidefinite matrix with unit diagonals~\cite{higham2002computing}. 
  \item $\mathbf{P}$ should be positive definite to have a unique Cholesky decomposition defined as $\mathbf{P} = \mathbf{L}^T \mathbf{L}$ where $\mathbf{L}$ is a lower triangular matrix with positive diagonal entries~\cite[p. 187]{golub}. 
\end{enumerate}
To ensure Property 1, we use the algorithm from~\cite[\S 3.2]{higham2002computing}, denoted NCM, to obtain the matrix $\mathbf{P}'$ as
\begin{equation}
\mathbf{P}' =\text{NCM}\left(\mathbf{P} \right),\label{eq:Q2}
\end{equation}
Furthermore, to satisfy Property 2, i.e., positive definiteness, we force the zero eigenvalues of $\mathbf{P}'$ to small positive values. For completeness, we describe the simplified skeleton of the algorithm from~\cite{higham2002computing} in Algorithm~\ref{alg:algo_NCM}. The algorithm searches for a correlation matrix that is closest to $\mathbf{P}$ in a weighted Frobenius norm. The output is asymptotically guaranteed to output the nearest correlation matrix to the input matrix~\cite{higham2002computing}. For details see~\cite[\S 3.2, p. 11]{higham2002computing}. The overall complexity of the procedure is $O(d^\omega)$, where $\omega < 2.38$ is the coefficient in the cost of multiplying two $d \times d$ matrices~\cite{omega}. Note that matrix multiplication is highly optimized in modern day computing languages.

\begin{algorithm*}[t]
\caption{\label{alg:algo_NCM} Algorithm to Obtain $\textrm{NCM}\left(\mathbf{P}\right)$}

\begin{enumerate}
\item Initialization: $\bigtriangleup\mathbf{S}_{0} \leftarrow \mathbf{0}$,$\mathbf{Y}_{0} \leftarrow \mathbf{P}$,
$k \leftarrow 1$.
\item While $k < \mathsf{iters}$, where $\mathsf{iters}$ is the maximum of iterations
\begin{enumerate}
\item $\mathbf{R}_{k}=\mathbf{Y}_{k-1}-\bigtriangleup\mathbf{S}_{k-1}$
\item Projection of $\mathbf{R}_{k}$ to a semi-positive definite matrix: 
$\mathbf{X}_{k} \leftarrow \mathbf{V}^{\textrm{T}}\textrm{diag}\left(\max\left\{ \mathbf{\Lambda},\mathbf{0}\right\} \right)\mathbf{V}$, where $\mathbf{V}$ and $\mathbf{\Lambda}$ contain the eigenvectors and the eigenvalues of $\mathbf{R}_{k}$, respectively, and $\textrm{diag}\left(\cdot\right)$ transforms a vector into a diagonal matrix.
\item $\bigtriangleup\mathbf{S}_{k} \leftarrow \mathbf{X}_{k}-\mathbf{R}_{k}$
\item Projection of $\mathbf{X}_{k}$ to a unit-diagonal matrix: $\mathbf{Y}_{k} \leftarrow \text{unitdiag} \left(\mathbf{X}_{k}\right)$, where $\text{unitdiag}(\cdot)$ fixes the diagonal of the input matrix to ones.
\item $k \leftarrow k + 1$.
\end{enumerate}
\item Output: $\mathbf{P}' \leftarrow \mathbf{Y}_{k}$.
\end{enumerate}
\end{algorithm*}

Finally, we generate records from a multivariate normal distribution using the well known method described in Algorithm~\ref{alg:algo_gen_multivariate_norm_distri}. The rationale of invoking Cholesky decomposition is to ensure that
\begin{equation*}
\mathbb{E}\left\{ \mathbf{Y}^{T}\mathbf{Y}\right\} =\mathbb{E}\left\{ \mathbf{L}^{T}\mathbf{Z}^{T} \mathbf{Z} \mathbf{L}\right\} = \mathbf{L}^{T}\mathbb{E}\left\{\mathbf{Z}^{T}\mathbf{Z}\right\} \mathbf{L}=\mathbf{L}^{T}\mathbf{L}=\mathbf{P}',
\end{equation*}
where we have used the fact that $\mathbb{E}\left\{ \mathbf{Z}^{T}\mathbf{Z}\right\} = \mathbf{I}$ because each record in
$\mathbf{Z}$ follows i.i.d. multivariate normal distribution. The output dataset $D'_{\mathsf{B}}$ is the final synthetic dataset.

\begin{algorithm}
\caption{\label{alg:algo_gen_multivariate_norm_distri} Generating Records Following a Multivariate Normal Distribution}

\begin{enumerate}
\item Generate $n$ records,
with the attributes in each record following i.i.d. standard normal distribution, i.e.,
\begin{equation*}
\mathbf{Z}=
\begin{bmatrix}
Z^{(1)}  &  Z^{(2)} & \cdots & Z^{(n)} 
\end{bmatrix}^T
\end{equation*}
where $Z^{(k)}$ is given by $\begin{bmatrix}
Z_{1}^{(k)} & Z_{2}^{(k)} & \cdots & Z_{d}^{(k)}
\end{bmatrix}^T$ and $Z_{i}^{(k)}$'s are i.i.d. standard normal random variables.
\item Compute $\mathbf{Y}$ as $\mathbf{Y} = \mathbf{Z} \mathbf{L}$ where $\mathbf{L}^{T}\mathbf{L} = \mathbf{P}'$ and $\mathbf{L}$ is obtained by the Cholesky decomposition as
$\mathbf{L}=\textrm{chol}\left(\mathbf{P}'\right)$.
\item From the matrix $\mathbf{Y} = \begin{bmatrix} Y^{(1)} & Y^{(2)} & \cdots & Y^{(n)} \end{bmatrix}^T$, 
%
%
map every element $Y_{i}^{(k)}$ in each record $Y^{(k)}$ to $\tilde{X}_{i}^{(k)}$ using Eq.~\ref{eq:RV_transformation}, where $i \in \left\{ 1,2,\ldots, d \right\}$.
\item Output the mapped data as $D'_{\mathsf{B}}$.
\end{enumerate}
\end{algorithm}

\subsection{Privacy of the Scheme}
Let $m' = m + \binom{m}{2}$,
where $m$ is the number of attributes in $D$.
Let $\epsilon < 1$,
say $\epsilon = 0.99$.
Fix a $\delta$.
We set each of the $\epsilon'_i$'s and $\epsilon''_{i, j}$'s for $i, j \in [m]$ to $\epsilon'$,
where $\epsilon'$ is such that it gives us the target $\epsilon$ through Theorem~\ref{the:comp} by setting $k = m'$.
According to the advanced composition theorem (Theorem~\ref{the:comp}), since each of our $m'$ mechanisms are $(\epsilon', 0)$ differentially private, the overall construction is $(\epsilon, \delta)$-differentially private. Privacy budget consumed over each of the $m'$ mechanisms is roughly $\frac{1}{\sqrt{m'}}$. Note that all the algorithms in Section~\ref{sub:cop-construct} do not require access to the original dataset and therefore privacy is not impacted due to the post-processing property (see Theorem~\ref{the:post-proc}).



\section{Experimental Evaluation and Utility Analysis}
\label{sec:exp}

\subsection{Query Class and Error Metrics}

As mentioned in Section~\ref{sec:intro} and explained in Section~\ref{subsec:algo_orig2Gauss_mapping}, our focus is on all one-way
marginal and two-way positive conjunction queries. In addition, we will also evaluate the performance of our method on the set of three-way positive conjunction queries to demonstrate that our method can also give well approximated answers to other types of queries. Thus, we use the following query class to
evaluate our method, $Q := Q_1 \cup Q_2 \cup Q_3$. Here, $Q_1$ is the set of one-way
marginal counting queries, which consists of queries $q$ specified as

\[
q(i, D_{\mathsf{B}}) = \sum_{k = 1}^n I \{ X_i^{(k)} = b \}
\]
where $i \in [d]$ and $b \in \{0, 1\}$. The class $Q_2$ is the set of positive two-way conjunctions and consists of queries $q$ specified as
\[
q(i, j, D_{\mathsf{B}}) = \sum_{k = 1}^n I \{ (X_i^{(k)} = 1) \cap  (X_j^{(k)} = 1) \}
\]
where $i, j \in [d], j \ne i$. We define $Q_{12} = Q_1 \cup Q_2$. The class $Q_3$ is the set of positive three-way conjunctions, and is defined analogously. Note that only those queries are included in $Q_2$ and $Q_3$ whose corresponding binary columns are from \emph{distinct} original columns in $D$. Answers to queries which evaluate at least two binary columns from the same original column in $D$ can be trivially fixed to zero; as these are ``structural zeroes.'' We assume this to be true for queries in the two aforementioned query sets from here onwards.
Our error metric of choice is the absolute error, which for a query $q \in Q$ is defined as
\[
{| q(D_{\mathsf{B}}) - q(D_{\mathsf{B}}') |}
\]
We have preferred this error metric over relative error (which scales the answers based on the original answer), since it makes it easier to compare results across different datasets and query types. For instance, in the case of relative error, a scaling factor is normally introduced, which is either a constant or a percentage of the number of rows in the dataset~\cite{wavelet, dp-copula}. The scaling factor is employed to not penalize queries with extremely small answers. However, the instantiation of the scaling factor is mostly a heuristic choice. 

For each of the datasets (described next), we shall evaluate the differences in answers to queries from $Q$ in
the remainder of this section. We shall be reporting $(\alpha, \beta)$-utility in the following way. We first sort the query answers in ascending order of error. We then fix a value of $\beta$, and report the maximum error and average error from the first $1 - \beta$ fraction of queries. We shall use values of $\beta = 0.05, 0.01$ and $0$. The errors returned then correspond to 95\%, 99\% and 100\% (overall error) of the queries, respectively. 

\subsection{Datasets and Parameters}
\label{sub:data-and-params}
We used three real-world datasets to evaluate our method. All three datasets contained a
mixture of ordinal (both discrete and continuous) and categorical attributes.

\begin{enumerate}
 \item \emph{Adult Dataset:} This is a publicly available dataset which is an extract from the 1994 US census information~\cite{adult}. There are 14 attributes in total which after pre-processing result in 194 binary attributes. There a total of 32,560 rows in this dataset (each belonging to an individual).
 \item \emph{{\dss} Dataset:} This dataset is a subset of a dataset obtained from the Department of Social Services (DSS), a department of the Government of Australia. The dataset contains transactions of social security payments. The subset contains 27 attributes, which result in 674 binary attributes. There are 5,240,260 rows in this dataset.
 \item \emph{Hospital Dataset:} This dataset is a subset of the hospital ratings dataset extracted from a national patient survey in the US.\footnote{See \url{https://data.medicare.gov/Hospital-Compare/Patient-survey-HCAHPS-Hospital/dgck-syfz}.} The dataset is among many other datasets gathered by the Centers for Medicare \& Medicaid Services (CMS), a federal agency in the US. We extracted (the first) 9 attributes and 10,000 rows for this dataset (resulting in 1,201 binary attributes). We mainly chose this dataset as an example of a dataset with highly correlated attributes.
\end{enumerate}


%
\paragraph*{{Parameter Values.}} We set $\delta = 2^{-30}$, following a similar value used in~\cite[\S 3, p. 5]{sec-sample}. This is well below ${n^{-1}}$ for the three datasets, where $n$ is the number of rows. We set the same privacy budget, $\epsilon'$, to compute each of the $m$ one-way marginals and $\binom{m}{2}$ two-way positive conjunctions. For the Adult, {\dss} and Hospital datasets we have $m = 14$, $m = 27$ and $m = 9$, respectively. We search for an $\epsilon'$ for the two datasets by setting $\delta = 2^{-30}$ and $k = m + \binom{m}{2}$ in Theorem~\ref{the:comp} which gives an overall $\epsilon$ of just under $1$. The resulting computation gives us $\epsilon' = 0.014782$ for Adult, $\epsilon' = 0.007791$ for {\dss}, and $\epsilon' = 0.022579$ for the Hospital dataset.

\subsection{Experimental Analysis}

To evaluate our method, we generate multiple synthetic datasets from each of
the three datasets. We will first evaluate the synthetic dataset generated
through the Gaussian copula (with no differential privacy) for each of the three datasets. This will be
followed by the evaluation of the differentially private version of this
synthetic dataset against the baseline which is
the application of the Laplace mechanism on the set of queries $Q$.
Note that this does not result in a synthetic dataset. For readability, we use
abbreviations for the different outputs. These are shown in Table~\ref{tab:not-mean}.


\begin{table*}
\centering
\begin{tabular}{r|l|c|c}
\multirow{ 2}{*}{Output} & \multirow{ 2}{*}{Mechanism} & Differential & \multirow{ 2}{*}{Synthetic} \tabularnewline
&&Privacy& \tabularnewline
 \hline\hline
cop & Gaussian copula with original correlation matrix & \xmark & \cmark \tabularnewline
cop-ID & Gaussian copula with identity correlation matrix & \xmark & \cmark \tabularnewline
cop-1 & Gaussian copula with correlation matrix of all 1's & \xmark & \cmark \tabularnewline
no-cor & Answers generated assuming no correlation  & \xmark & \xmark \tabularnewline
Lap & Laplace mechanism  & \cmark & \xmark \tabularnewline
dpc & Gaussian copula with diff. private correlation matrix & \cmark & \cmark \tabularnewline
\end{tabular}
\caption{Notation used for outputs from different mechanisms. Note that not all of them are synthetic datasets or differentially private.}
\label{tab:not-mean}
\end{table*}

\subsubsection{Error due to Gaussian Copula without Differential Privacy}
\label{subsub:cop-no-dp}
We first isolate and quantify query errors from a synthetic dataset obtained directly through the Gaussian copula, i.e., without any differentially private noise added to the one-way marginals and the correlation matrix (see Figure~\ref{fig:flowchart}). Since generating synthetic datasets through the copula is an inherently random process, this itself may be a source of error. We denote such a dataset by ``cop.'' Thus, Adult cop, DSS cop and Hospital cop are the ``cop'' versions of the corresponding datasets. We restrict ourselves to the query set $Q_{12}$ and compare error from cop against three other outputs:
\begin{itemize}
	\item[cop-ID: ] A synthetic dataset obtained by replacing the correlation matrix $\mathbf{P}'$ with the identity matrix. Evaluating against this dataset will show whether cop performs better than a trivial mechanism which assumes all binary attributes to be uncorrelated. Note that this mainly effects answers to $Q_2$, and not the one-way marginals $Q_1$.
	\item[cop-1: ] Another synthetic dataset obtained by replacing the correlation matrix $\mathbf{P}'$ with the matrix $\mathbf{1}$ of all ones. This serves as other extreme where all attributes are assumed to be positively correlated. Once again, this is to compare answers from $Q_2$.
	\item[no-cor: ] For the query class $Q_2$, we obtain a set of random answers which are computed by simply multiplying the means $\mu_1$ and $\mu_2$ of two binary attributes. This is the same as simulating two-way positive conjunctions of uncorrelated attributes. This should have an error distribution similar to the answers on $Q_2$ obtained from cop-ID. 
\end{itemize}

\paragraph*{Results.}
Figure~\ref{fig:cop:all} shows the CDF of the absolute error on the query set $Q_{12}$ from different outputs from the three datasets. Looking first at the results on $Q_1$ (top row in the figure), we see that for all three datasets cop has low absolute error, yielding $(\alpha, \beta)$-utility of $(149, 0.01)$ for Adult, $(170, 0.01)$ for DSS and $(28, 0.01)$ for the Hospital dataset in terms of max-absolute error. The datasets cop-ID and cop-1 exhibit similar utility for all three datasets. This is not surprising since $Q_1$ contains one-way marginals, whose accuracy is more impacted by the inverse transforms (Eq.~\ref{eq:RV_transformation}) rather than the correlation matrix. Answers on $Q_2$ are more intriguing (bottom row of Figure~\ref{fig:cop:all}). First note that the utility from cop is once again good for all three input datasets with 99\% of the queries having a maximum absolute error of $353$ for Adult, 83 for DSS and 6 for the Hospital dataset. The cop-1 outputs have poorer utility. However, interestingly, cop-ID and no-cor outputs yield utility very similar to cop. We discuss this in more detail next.

\begin{figure}[!tbh]
\centering
\subfigure[Set $Q_1$ (one-way marginals)]{\label{fig:cop:one-way:adult}
\includegraphics[width=0.3\textwidth]{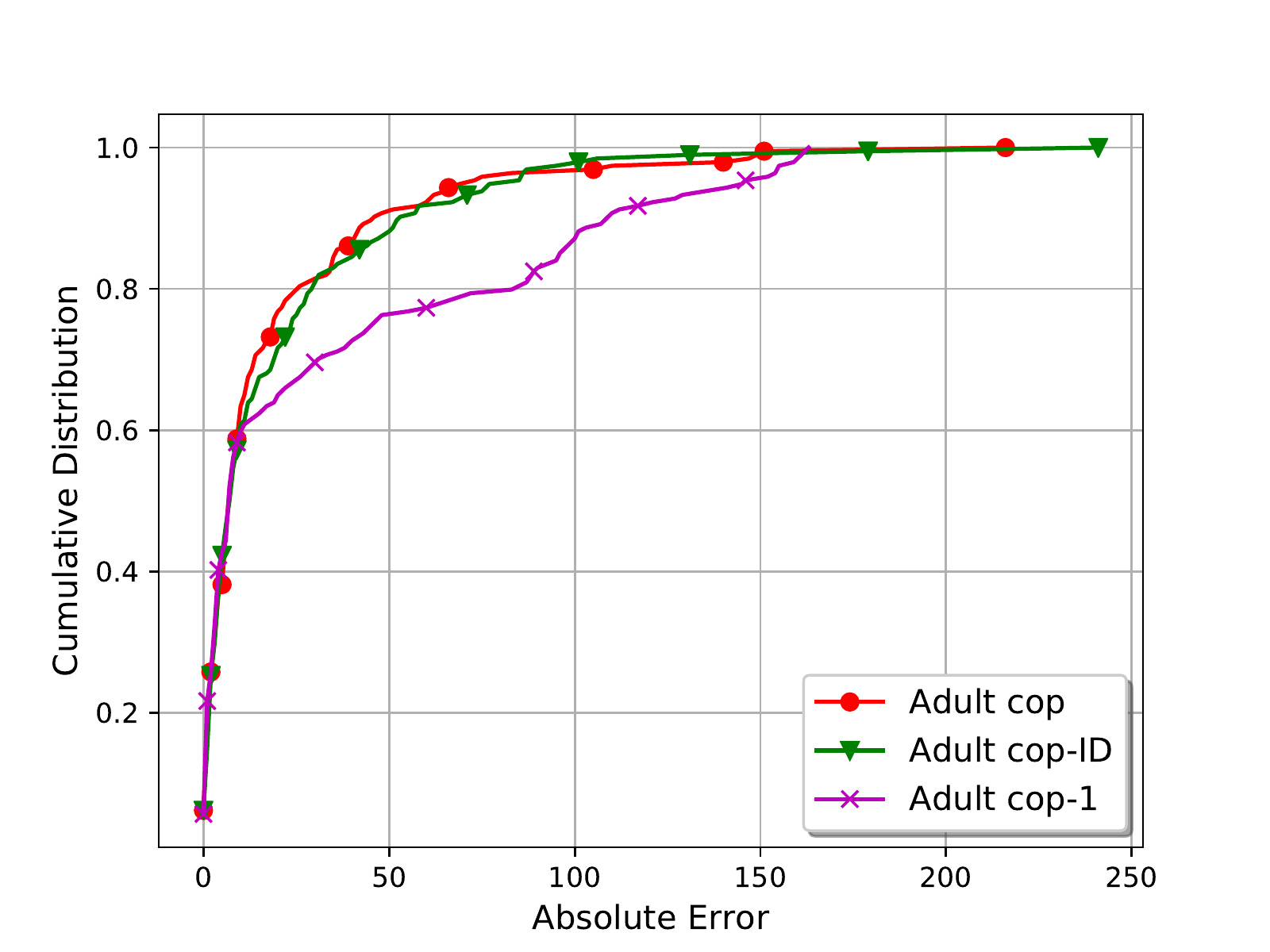}}
~
\subfigure[Set $Q_1$ (one-way marginals)]{\label{fig:cop:one-way:dss}
\includegraphics[width=0.3\textwidth]{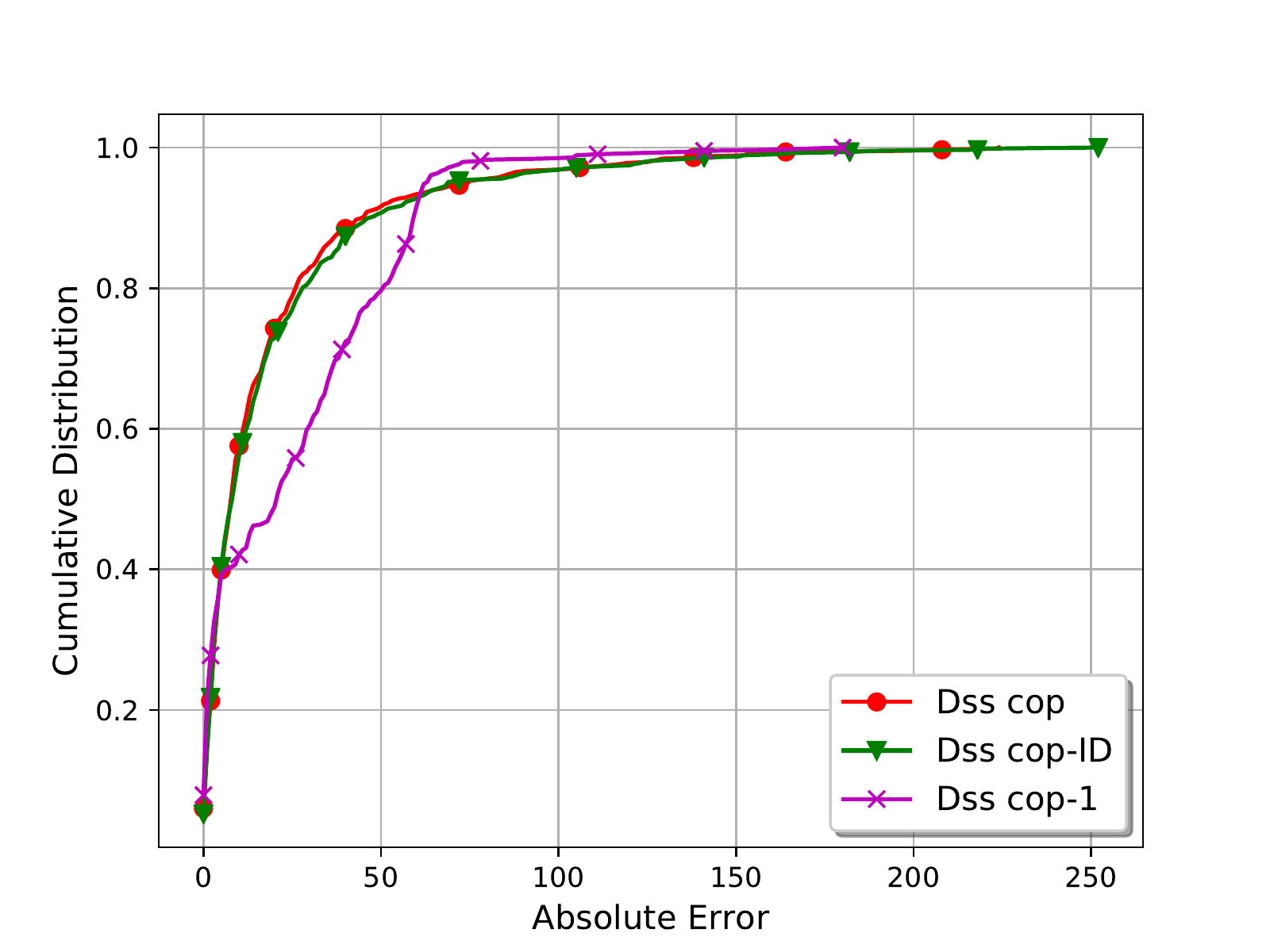}}
~
\subfigure[Set $Q_1$ (one-way marginals)]{\label{fig:cop:one-way:hos}
\includegraphics[width=0.3\textwidth]{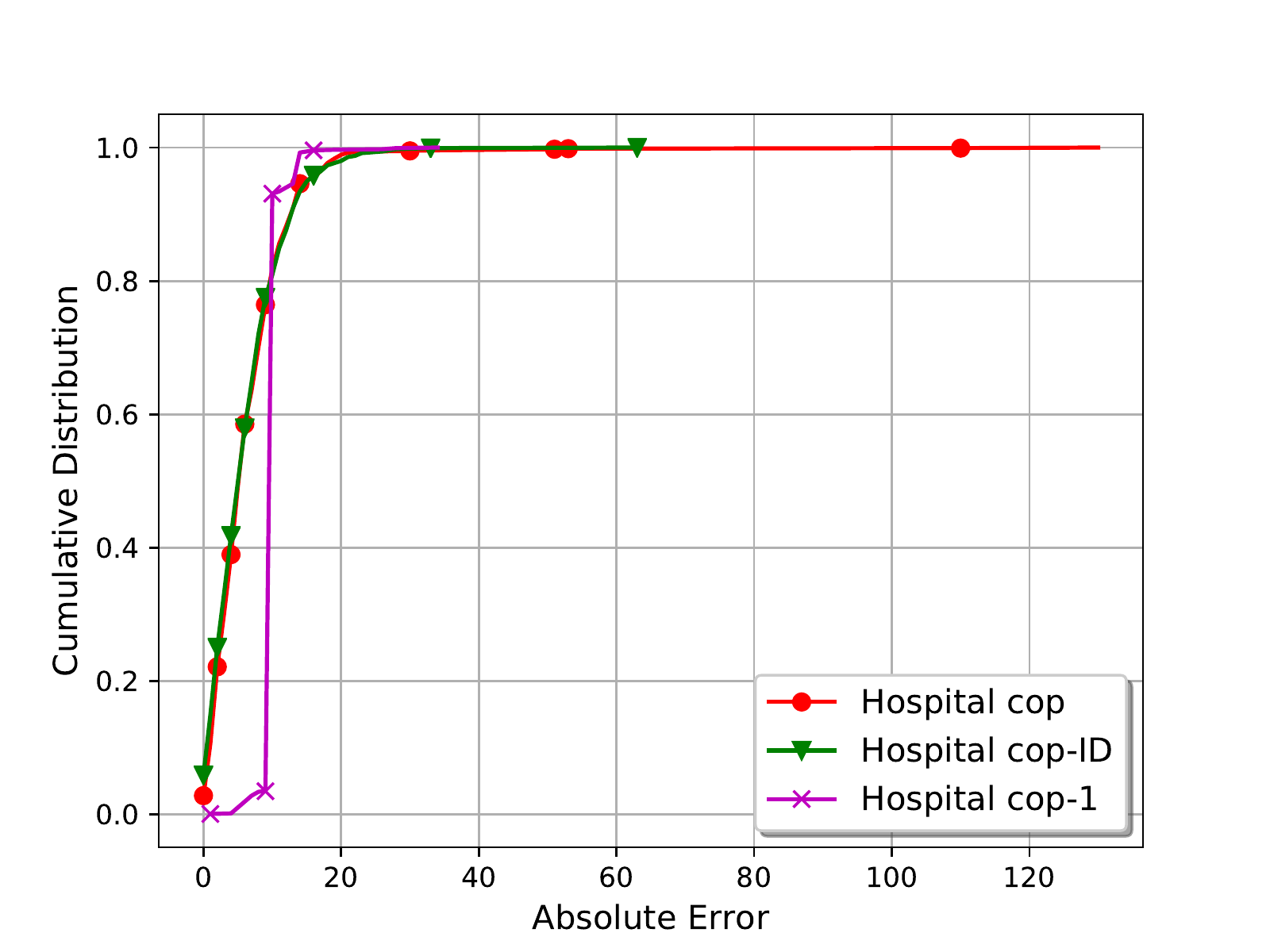}}

\subfigure[Set $Q_2$ (two-way conjunctions)]{\label{fig:cop:two-way:adult}
  \includegraphics[width=0.3\textwidth]{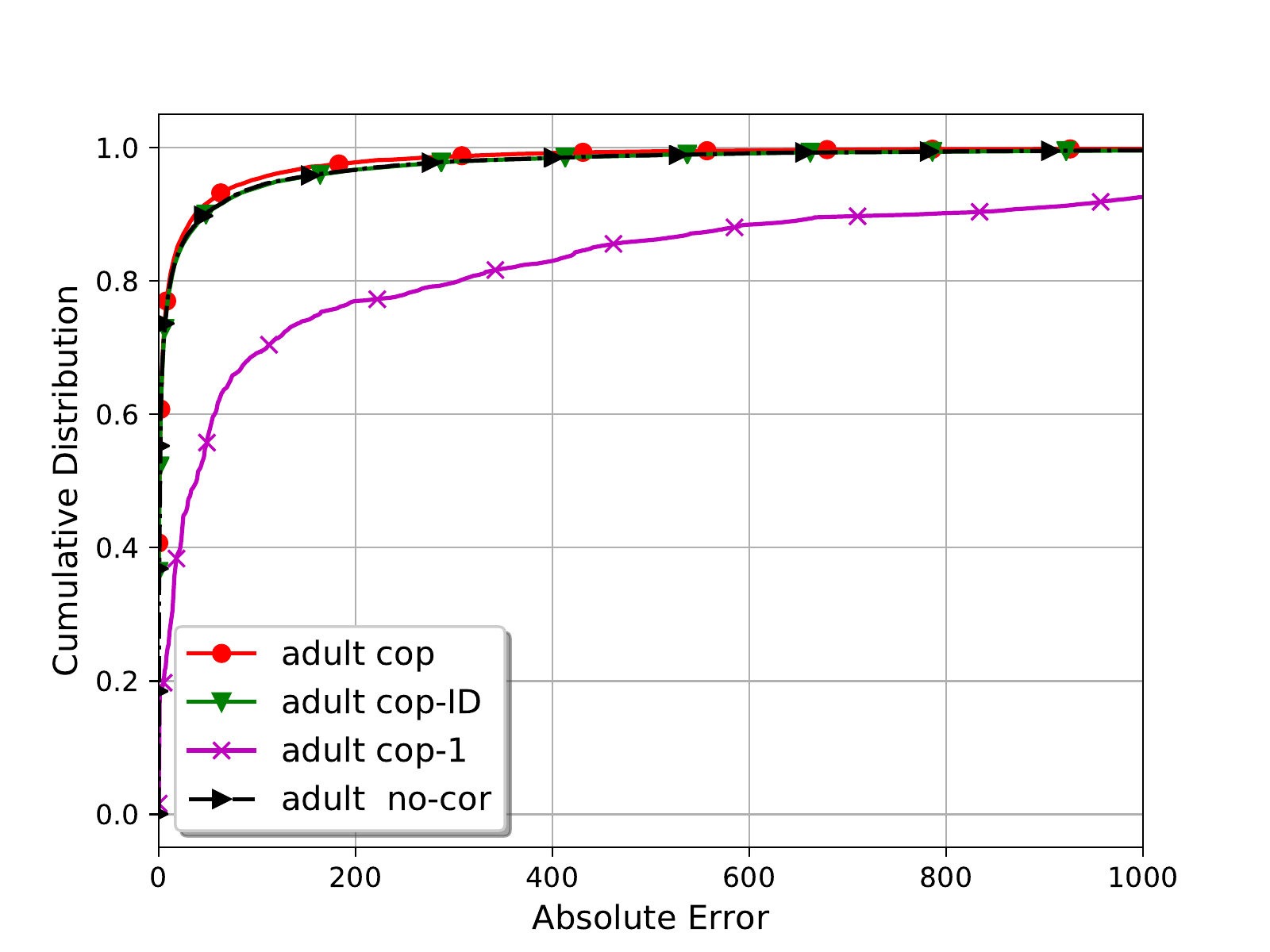}}
~
\subfigure[Set $Q_2$ (two-way conjunctions)]{\label{fig:cop:two-way:dss}
  \includegraphics[width=0.3\textwidth]{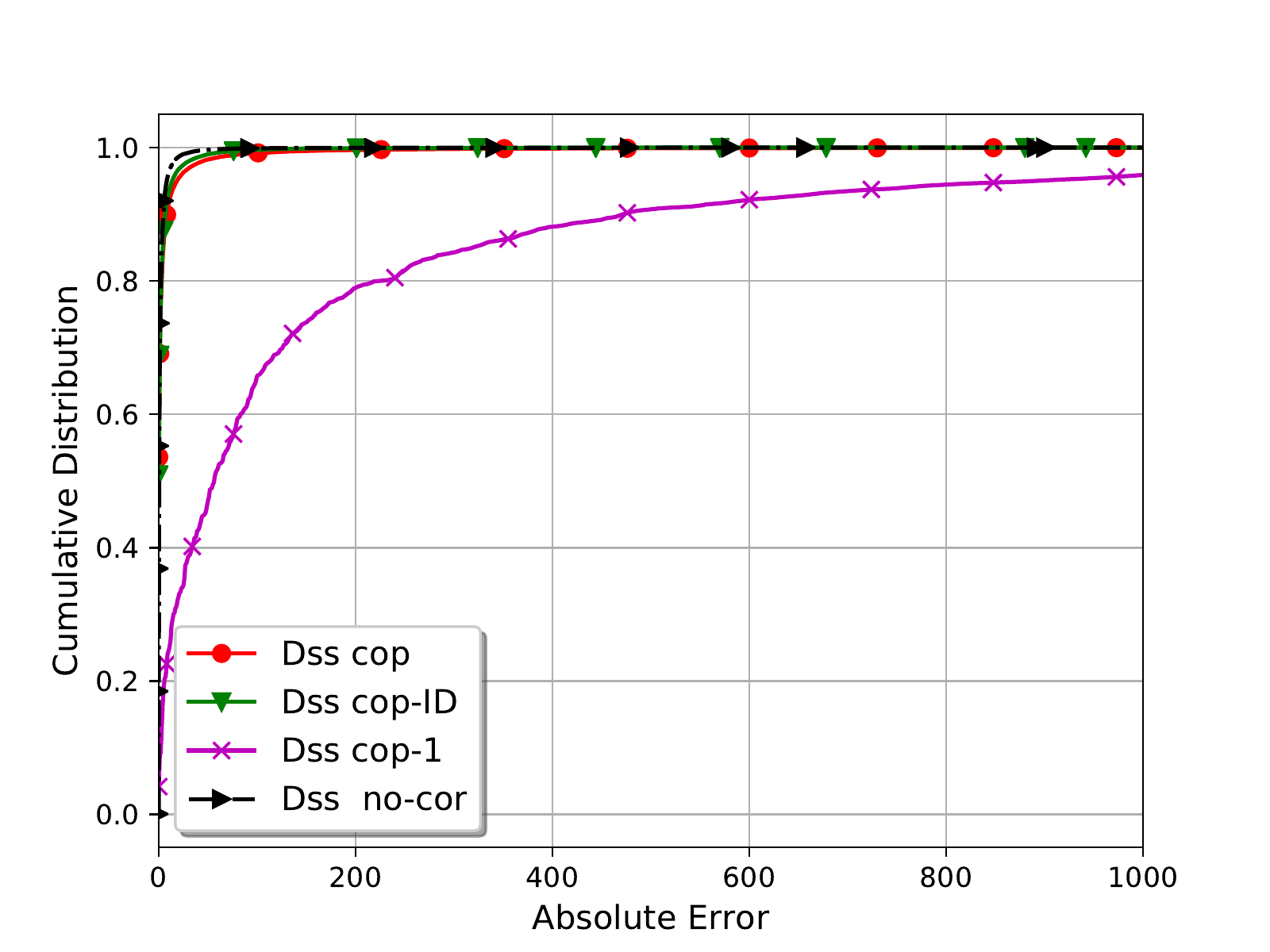}}
~
\subfigure[Set $Q_2$ (two-way conjunctions)]{\label{fig:cop:two-way:hos}
  \includegraphics[width=0.3\textwidth]{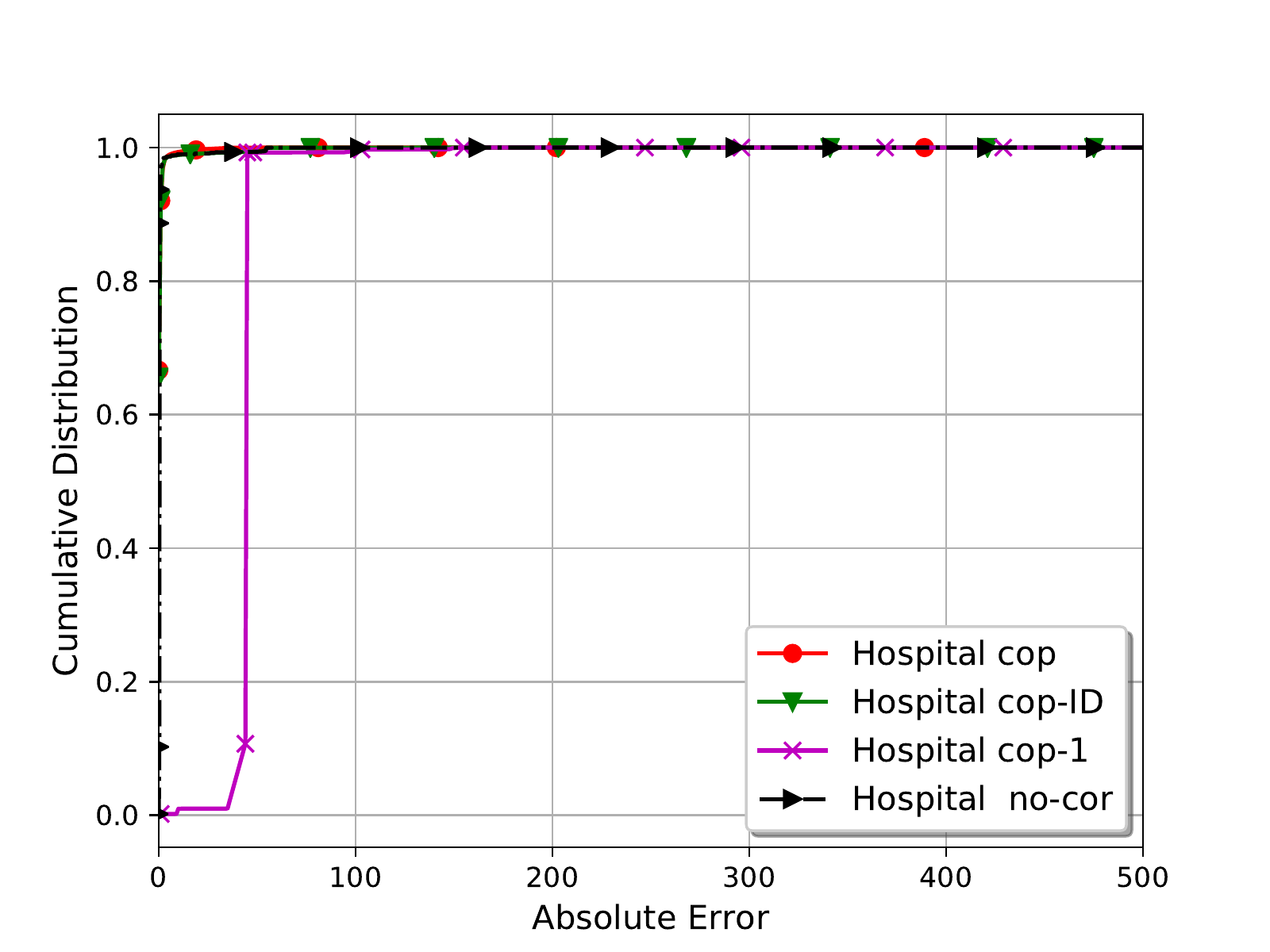}}  
\caption{Relative error over $Q_{12}$ on synthetic versions of the Adult, DSS and Hospital datasets (without differential privacy).}
\label{fig:cop:all}
\end{figure}

%

\paragraph*{{Separating Low and High Correlations.}}
There are two possible reasons why cop-ID and no-cor outputs perform close to cop outputs on the set $Q_2$: (a) our method does
not perform better than random (when it comes to $Q_2$), or (b) uncorrelated attributes dominate the dataset thus overwhelming the distribution of errors. The second reason is also evidenced by the fact that the cop-1 dataset, using a correlation matrix of all ones, performs worse than the other three outputs, indicating that highly correlated attributes are rare in the three datasets.

To further ascertain which of the two reasons is true, we separate the error results on the set $Q_2$ into two parts: a set of (binary) attribute pairs having high correlations (positive or negative), and another with low correlations. If our method performs better than random, we would expect the error from cop to be lower on the first set when compared to cop-ID and no-cor, while at least comparable to the two on the second set. We use the Hospital dataset for this analysis as it was the dataset with the most highly correlated attributes. There are a total of 626,491 pairs of binary attributes (as mentioned before, we ignore binary attribute pairs that correspond to different attribute values on the same attribute in the original dataset). Out of these, only 3,355 have an (absolute) Pearson correlation coefficient $|r| \ge 0.5$. Thus, an overwhelming 99.46\% of pairs have $|r| < 0.5$. This shows that the dataset does indeed contain a high number of low-correlated attributes, which partially explains similar error profile of cop, cop-ID and no-cor.

Figure~\ref{fig:cop:hos:corr} shows this breakdown. The error on the set with $|r| < 0.5$ is very similar for cop, cop-ID and no-cor (Figure~\ref{fig:cop:low-cor:hos}). Looking at Figure~\ref{fig:cop:high-cor:hos}, for the set with $|r| \ge 0.5$, on the other hand, we note that cop outperforms both cop-ID and no-cor. Also note that cop-1 is similar in performance to our method. This is understandable since cop-1 uses a correlation matrix with all ones, and hence is expected to perform well on highly correlated pairs. 
This indicates that our method outperforms cop-ID and no-cor. We conclude that the \emph{apparent} similarity between cop, cop-ID and no-cor on the set $Q_2$ is due to an artefact of some real-world datasets which may have a high number of uncorrelated attributes; when the results are analyzed separately, our method is superior. We remark that we arrived at the same conclusion for the Adult and {\dss} datasets, but omit the results due to repetition.

\begin{figure*}[!tbh]
\centering
\subfigure[$|r| <0.5$]{\label{fig:cop:low-cor:hos}
\includegraphics[width=0.32\textwidth]{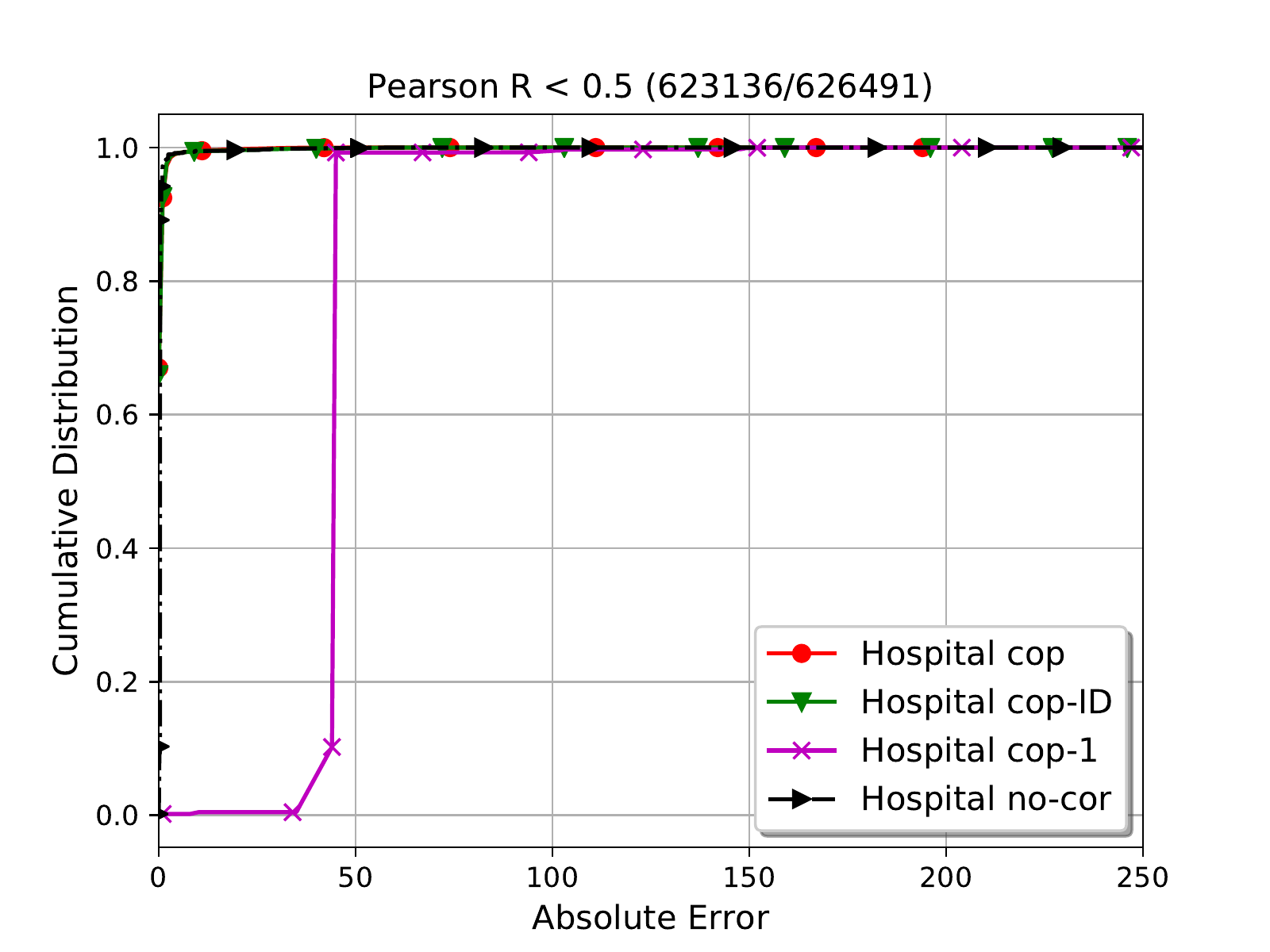}}
\subfigure[$|r| \geq0.5$]{\label{fig:cop:high-cor:hos}
\includegraphics[width=0.32\textwidth]{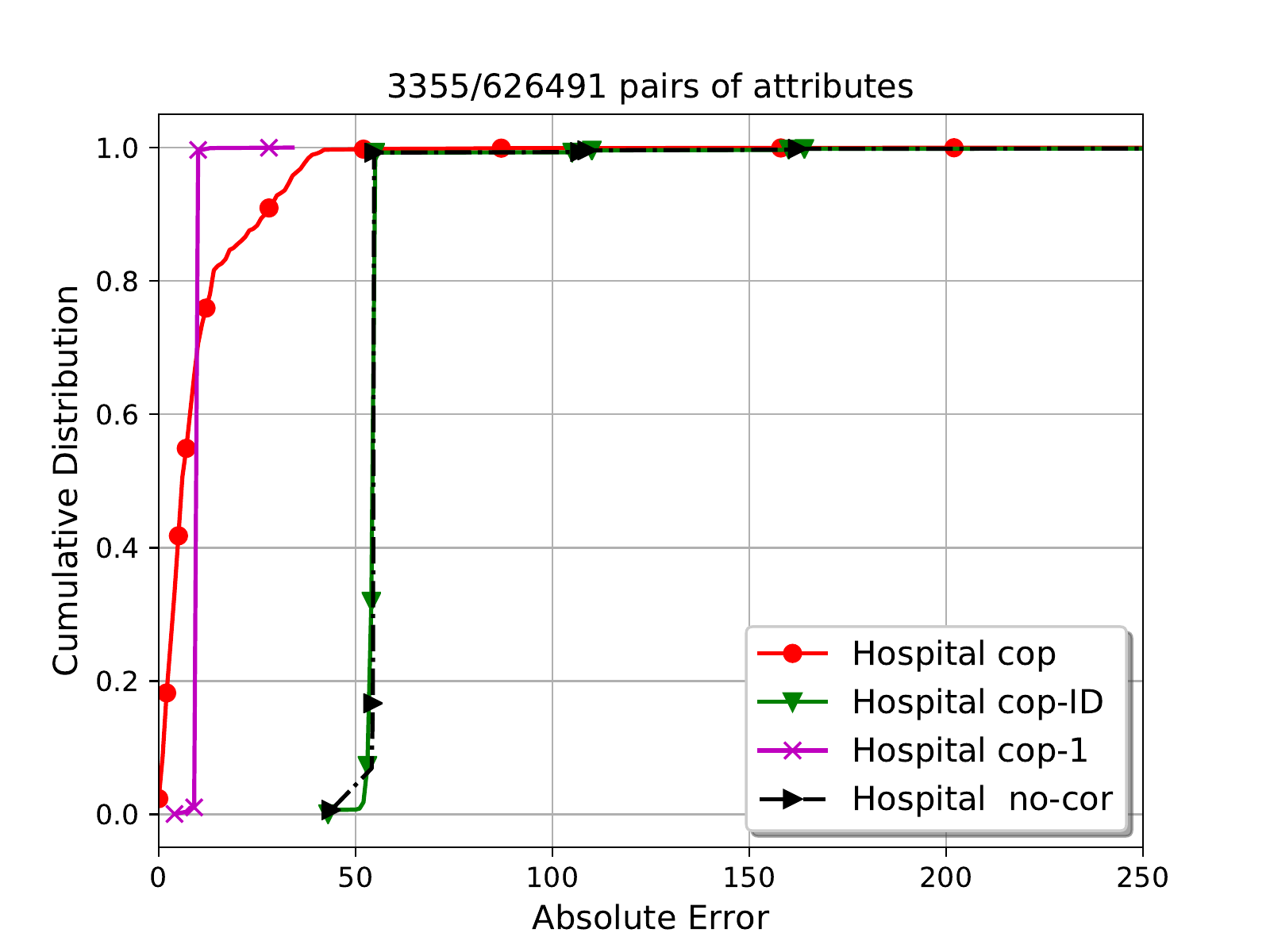}}
\caption{Absolute error over $Q_2$ on synthetic versions of the Hospital dataset (without differential privacy) for different values of the absolute Pearson correlation coefficient $|r|$ between pairs of attributes.}
\label{fig:cop:hos:corr}
\end{figure*}



\subsubsection{Error due to Differentially Private Gaussian Copula}
Having established that the error due to Gaussian copula is small, we now turn to the complete version of our method, i.e., with differential privacy. For this section, we are interested in three outputs: (a) dpc, i.e., the synthetic dataset obtained through our differentially private Gaussian copula method, (b) cop, i.e., the synthetic dataset via our Gaussian copula method without differential privacy, and (c) Lap, i.e., a set of answers obtained by adding independent Laplace noise to the answers to the queries in $Q_{12}$. Note that Lap is not a synthetic dataset. We set the same $\epsilon'$ for Lap, as we did for dpc. 

\paragraph*{{Results.}}
Figure~\ref{fig:dpcop} shows the absolute error CDF on the query set $Q_{12}$ for cop, dpc and Lap versions constructed from the three datasets. For the set $Q_1$ (top row in figure), we see that cop outperforms both dpc and Lap. This is due to the fact that for privacy, a higher amount of noise is required. Crucially, our method does not introduce further error over the Laplace mechanism, as is indicated by the similarity of the curves corresponding to dpc and Lap. Interestingly, the results for $Q_2$ show that dpc outperforms independent Laplace noise (bottom row of Figure~\ref{fig:dpcop}). While the error from dpc is still higher than cop, it is closer to it than the error due to Lap. This indicates that for the majority of the queries, our method applies less noise than Lap. 

However, in some cases, for a small percentage of queries Lap adds less noise than our mechanism. This is clear from Table~\ref{tab:alpha-values}, where we show the maximum and average error\footnote{Rounded to the nearest integer.} from 95\%, 99\% and 100\% percent of the queries from $Q_{12}$ across dpc and Lap versions of the Adult, {\dss} and Hospital datasets. The error profiles of the dpc and Lap variants are similar for the query set $Q_1$. For the set $Q_2$, we can see that Lap only outperforms our method for the Adult and {\dss} datasets if we consider the maximum absolute error across all queries. On the other hand our method outperforms Lap if we consider 95\% and 99\% of queries. Thus, for less than 1\% of queries, the dpc versions of Adult and {\dss} exhibits less utility than Lap, while being similar to the latter for the Hospital dataset.

\begin{figure*}[!tbh]
\centering
\subfigure[Set $Q_1$ (Adult)]{\label{fig:dpcop:one-way:adult}
\includegraphics[width=0.31\textwidth]{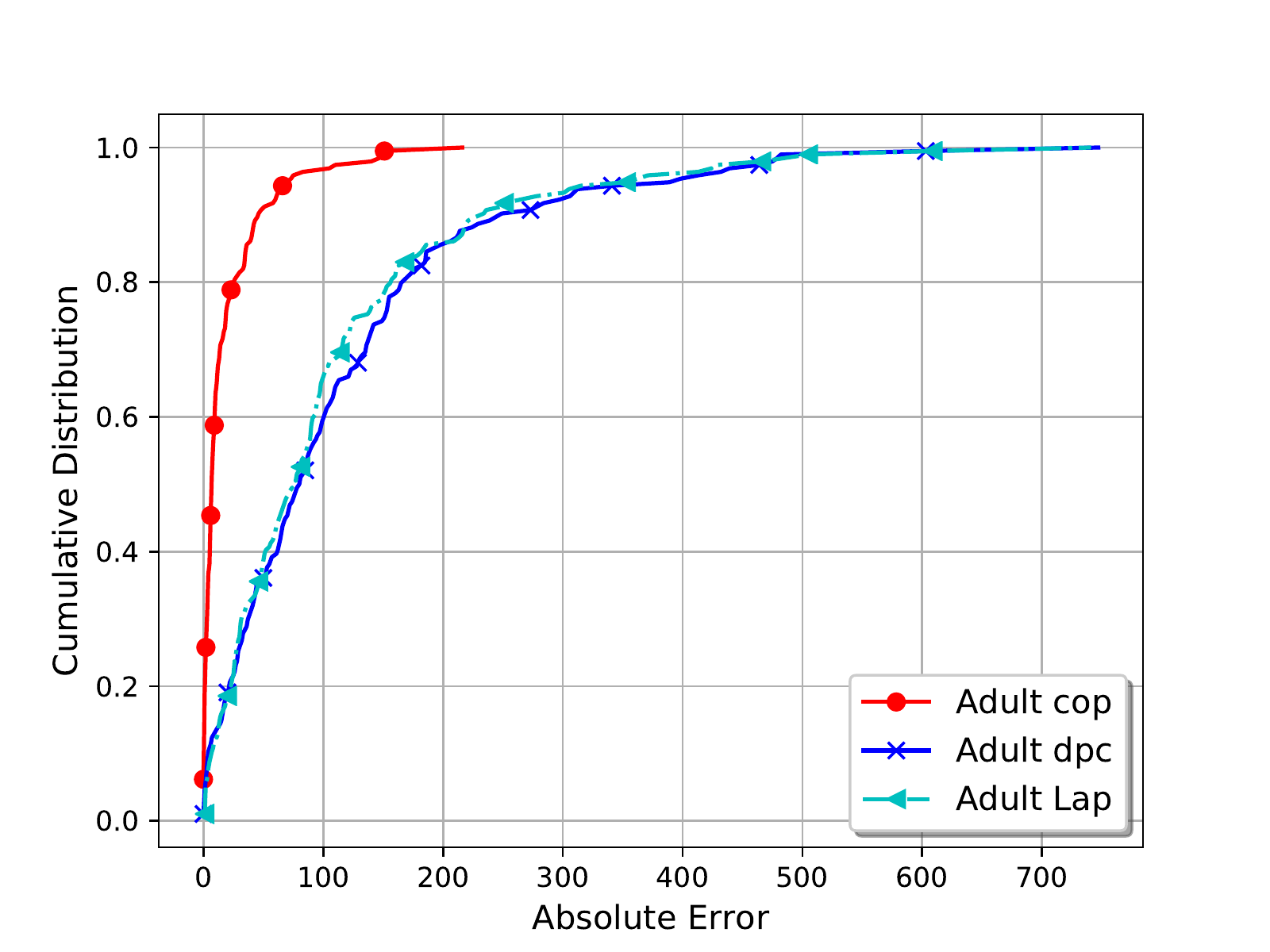}}
\subfigure[Set $Q_1$ ({\dss})]{\label{fig:dpcop:one-way:dss}
\includegraphics[width=0.31\textwidth]{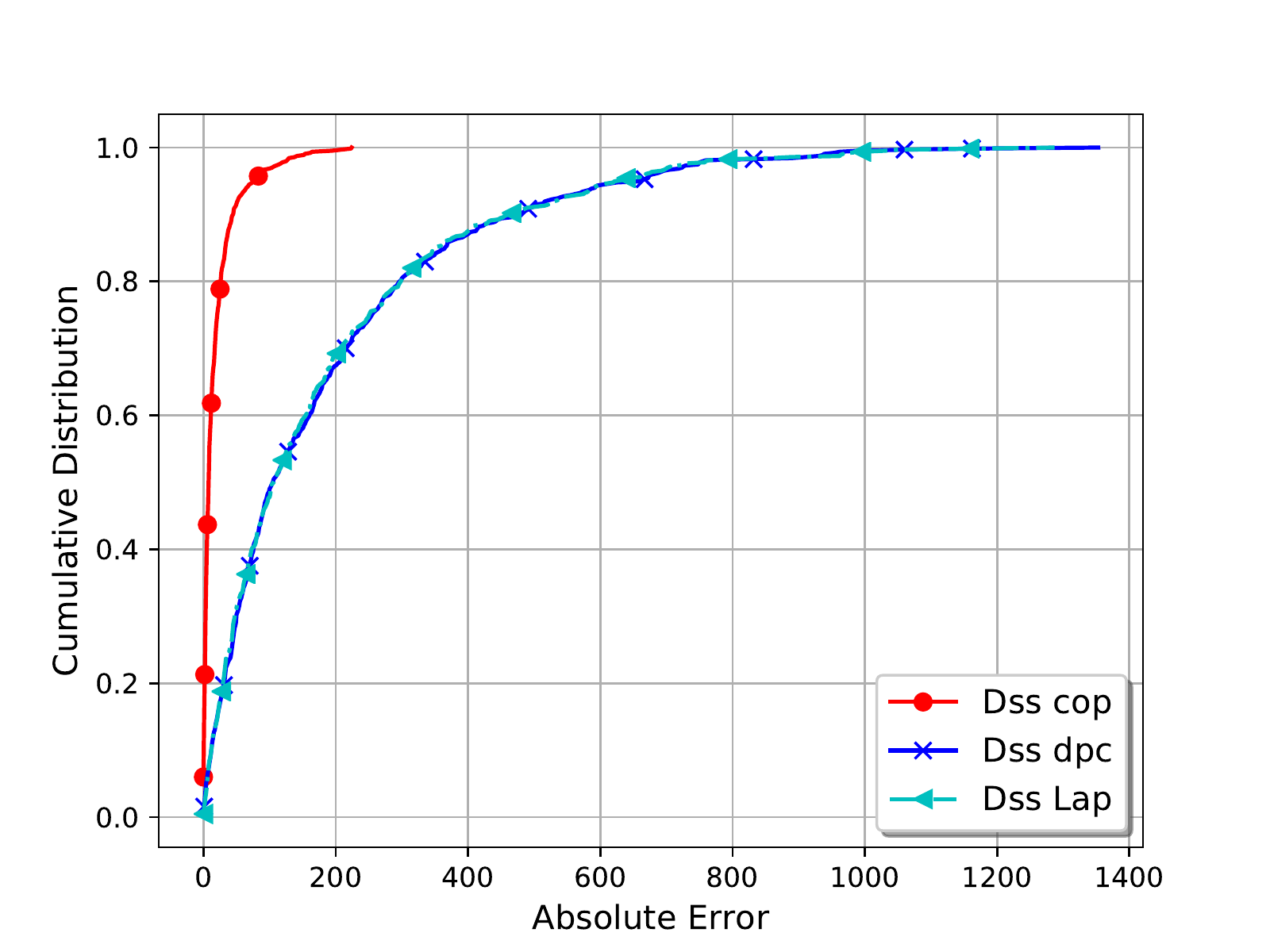}}
\subfigure[Set $Q_1$ (Hopspital)]{\label{fig:dpcop:one-way:hos}
\includegraphics[width=0.31\textwidth]{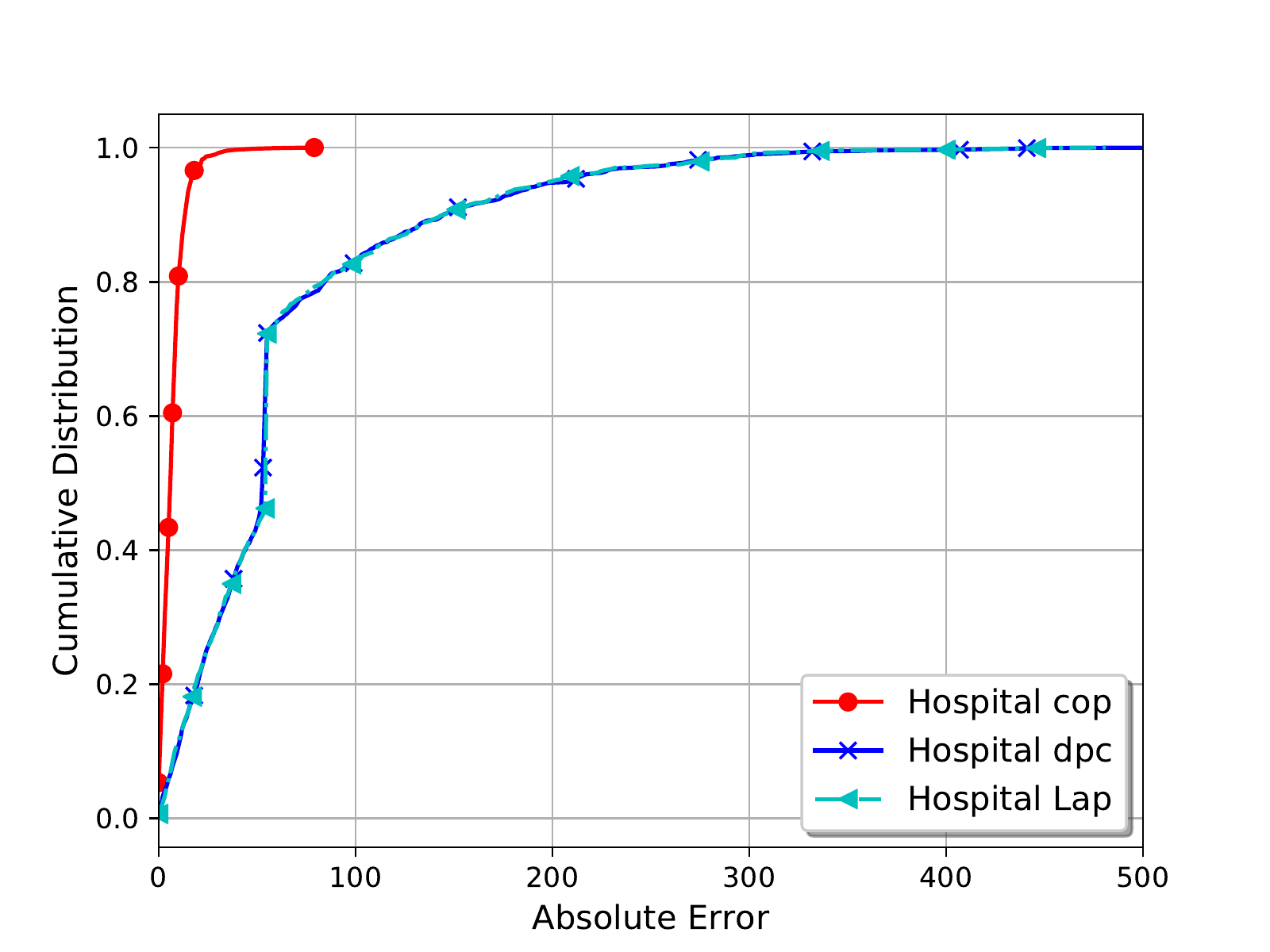}}

\subfigure[Set $Q_2$ (Adult)]{\label{fig:dpcop:two-way:adult}
  \includegraphics[width=0.31\textwidth]{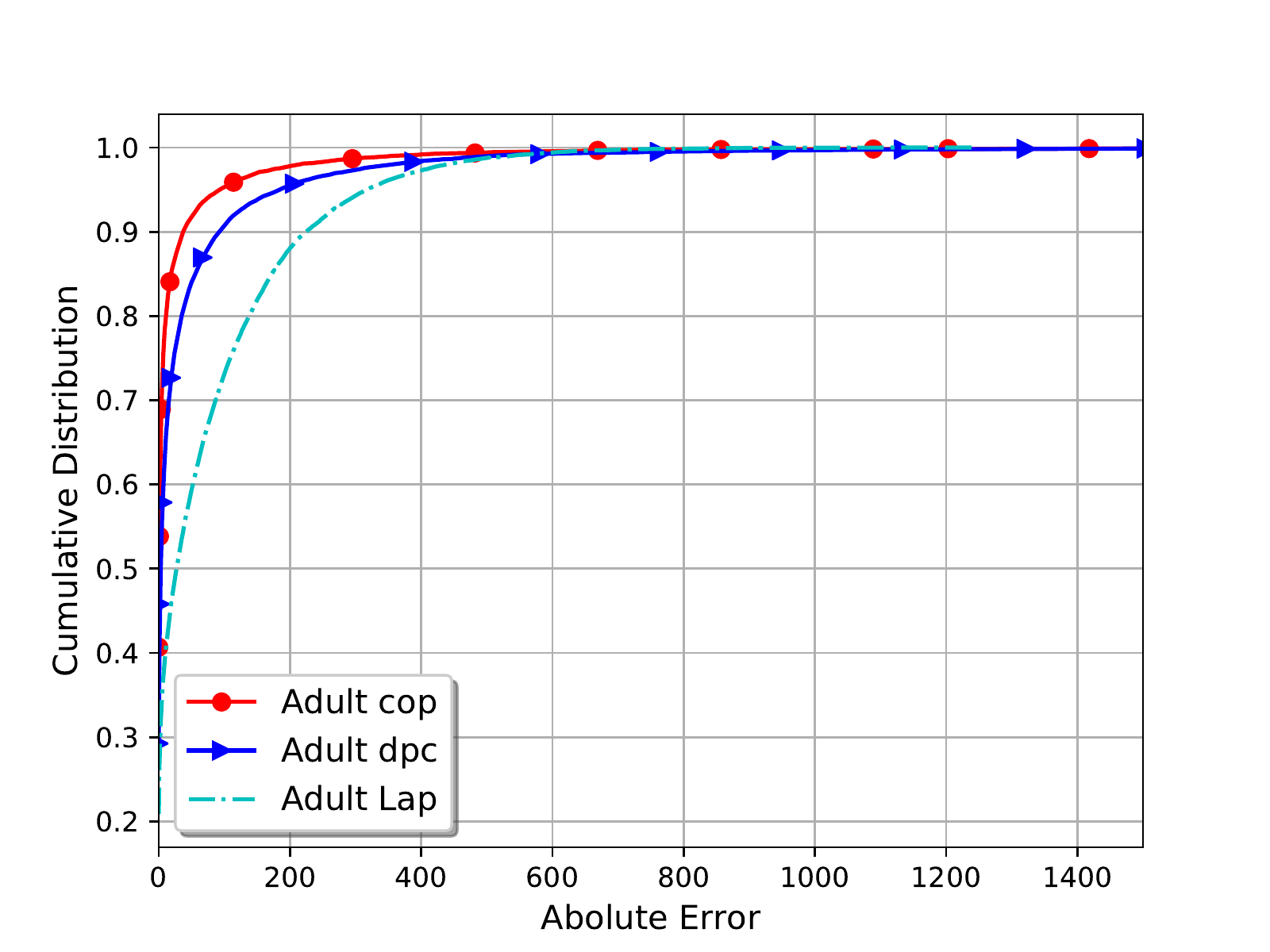}}
\subfigure[Set $Q_2$ ({\dss})]{\label{fig:dpcop:two-way:dss}
  \includegraphics[width=0.31\textwidth]{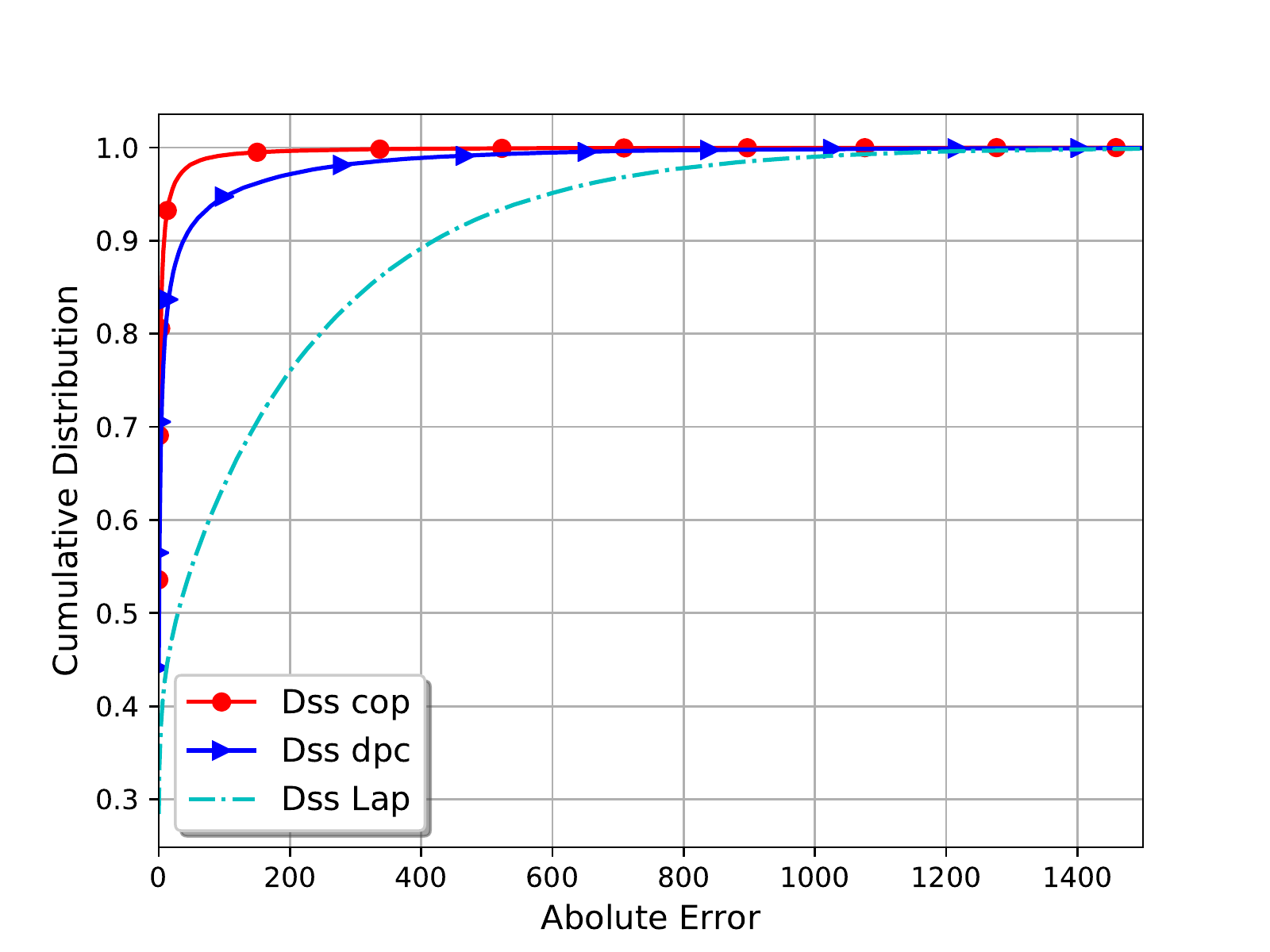}}
  \subfigure[Set $Q_2$ (Hopsital)]{\label{fig:dpcop:two-way:hos}
  \includegraphics[width=0.31\textwidth]{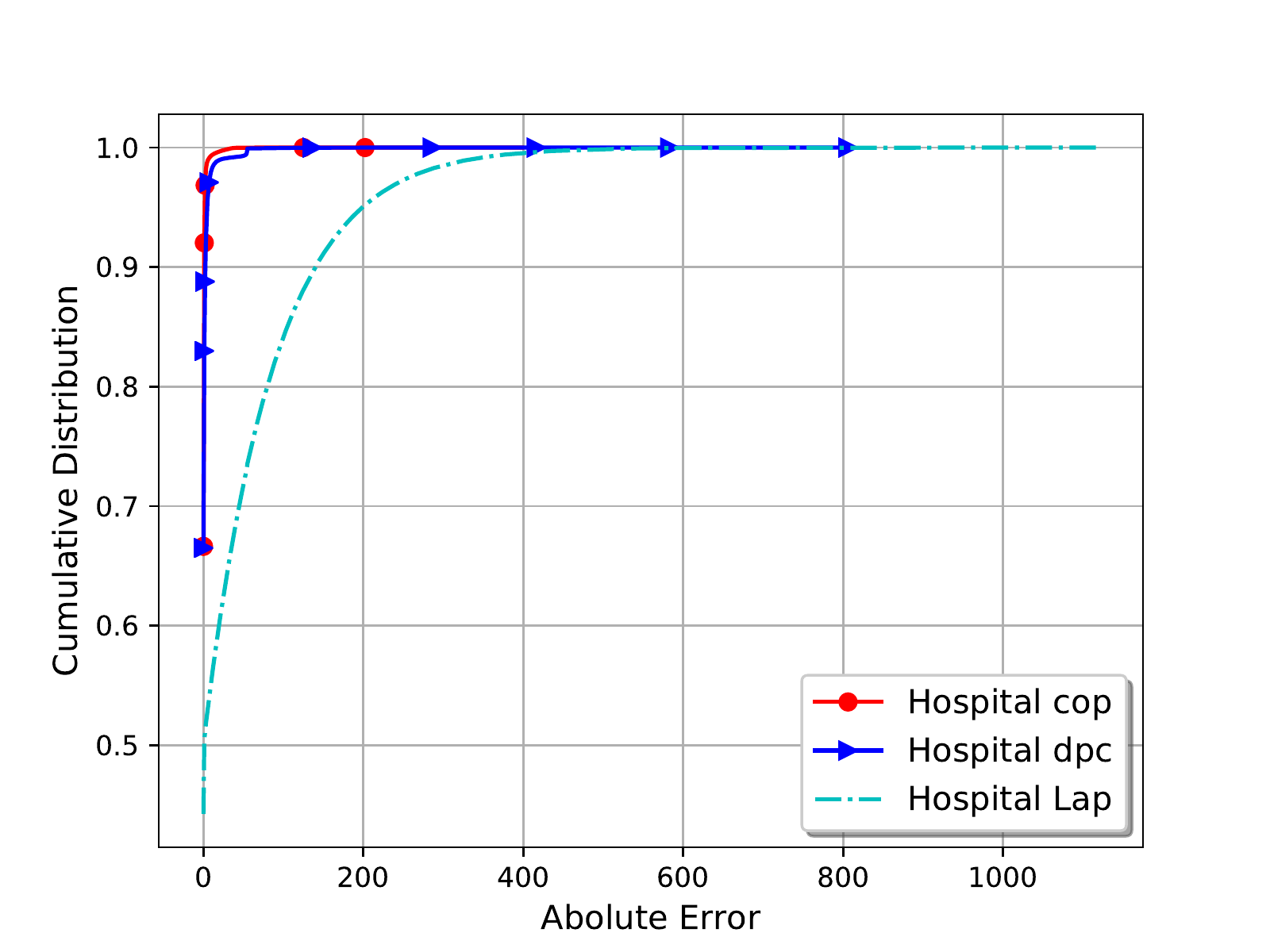}}
\caption{Absolute error over $Q_{12}$ of one-way marginals (set $Q_1$) and two-way positive conjunctions (set $Q_2$) on synthetic versions of the Adult (left), {\dss} (middle) and Hospital (right) datasets with differential privacy.}
\label{fig:dpcop}
\end{figure*}

\begin{table*}
\centering
\resizebox{\textwidth}{!}{
\begin{tabular}{c|c|c|c|c|c|c|c|c|c|c|c|c|c|c|c|c|c|c}
\multirow{3}{*}{Mechanism} & \multicolumn{6}{c}{Adult} & \multicolumn{6}{|c}{{\dss}}  & \multicolumn{6}{|c}{Hospital}  \\
\cline{2-19}
& \multicolumn{2}{c}{95\%} & \multicolumn{2}{|c}{99\%} & \multicolumn{2}{|c}{100\%} & \multicolumn{2}{|c}{95\%} & \multicolumn{2}{|c}{99\%} & \multicolumn{2}{|c}{100\%} & \multicolumn{2}{|c}{95\%} & \multicolumn{2}{|c}{99\%} & \multicolumn{2}{|c}{100\%} \\
\cline{2-19}
& ave & max & ave & max & ave & max & ave & max & ave & max & ave & max & ave & max & ave & max & ave & max \\
 \hline\hline
dpc (one-way) & 92 & 389 & 107 & 482 & 106 & 773 		&  151 & 657 & 175 & 982 & 185 & 1267 & 53 & 208 & 61 & 302 & 65 & 622 \\
Lap (one-way) & 85 & 353 & 99 & 505 & 105 & 741  		&  149 & 621 & 172 & 963 & 182 & 1287 & 52 & 199 & 61 & 299 & 61 & 887 \\
\hline
dpc (two-way) & \cellcolor{ikgreen} 18 & \cellcolor{ikgreen} 184 & \cellcolor{ikgreen} 29 & 504 & \cellcolor{ikgreen} 38 & \cellcolor{ikred} 4788 		&   \cellcolor{ikgreen}   7 & \cellcolor{ikgreen} 108 & \cellcolor{ikgreen} 15 & \cellcolor{ikgreen} 426 & \cellcolor{ikgreen} 23 & \cellcolor{ikred} 5744   & \cellcolor{ikgreen}  1 & \cellcolor{ikgreen} 5 & \cellcolor{ikgreen} 1 & \cellcolor{ikgreen} 22 & \cellcolor{ikgreen} 1 & 836 \\
Lap (two-way) & \cellcolor{ikgreen} 58 & \cellcolor{ikgreen} 317 & \cellcolor{ikgreen} 72 & 536 & \cellcolor{ikgreen} 78 & \cellcolor{ikred} 1246 		&   \cellcolor{ikgreen}  99 & \cellcolor{ikgreen} 594 & \cellcolor{ikgreen} 125 & \cellcolor{ikgreen} 995 & \cellcolor{ikgreen} 136 & \cellcolor{ikred} 2947 & \cellcolor{ikgreen} 31 & \cellcolor{ikgreen} 199 & \cellcolor{ikgreen} 40 & \cellcolor{ikgreen} 333 & \cellcolor{ikgreen} 43 & 1119\\
\hline
dpc (three-way) & \cellcolor{ikgreen} 12 & \cellcolor{ikgreen} 120 & \cellcolor{ikgreen} 20 & \cellcolor{ikgreen} 408 & \cellcolor{ikgreen} 28 & \cellcolor{ikred} 6148   	& 	\cellcolor{ikgreen}  9 & \cellcolor{ikgreen} 98 & \cellcolor{ikgreen} 16 & \cellcolor{ikgreen} 372 & \cellcolor{ikgreen} 36 & 9238      & \cellcolor{ikgreen} 1 & \cellcolor{ikgreen} 3 & \cellcolor{ikgreen} 3 & \cellcolor{ikgreen} 55 & \cellcolor{ikgreen} 3 & \cellcolor{ikgreen} 616 \\
Lap (three-way) & \cellcolor{ikgreen} 102 & \cellcolor{ikgreen} 591 & \cellcolor{ikgreen} 128 & \cellcolor{ikgreen} 1002 & \cellcolor{ikgreen} 138 & \cellcolor{ikred} 2508 &  \cellcolor{ikgreen} 268 & \cellcolor{ikgreen} 1651 & \cellcolor{ikgreen} 342 & \cellcolor{ikgreen} 2805 & \cellcolor{ikgreen} 726 & 9562 & \cellcolor{ikgreen} 45 & \cellcolor{ikgreen} 281 & \cellcolor{ikgreen} 57 & \cellcolor{ikgreen} 477 & \cellcolor{ikgreen} 63 & \cellcolor{ikgreen} 1361 
\end{tabular}
}
\caption{Absolute error $\alpha$ of the dpc and Lap mechanisms on the three datasets. The columns show average and maximum values of $\alpha$ for 95\%, 99\% and 100\% of the queries (corresponding to $\beta = 0.05, 0.01$ and $0.00$, respectively). \textcolor{ikgreen}{\protect\rule[0pt]{2mm}{2mm}} indicates our method significantly outperforms Lap; \textcolor{ikred}{\protect\rule[0pt]{2mm}{2mm}} indicates Lap significantly outperforms our method; significance is defined as an error ratio of approximately 2 or more.}
\label{tab:alpha-values}
\end{table*}

\subsubsection{Results on Three-Way Conjunctions}
Even though our method is expected to perform best on the query set $Q_{12}$, we show that the method performs well on other types of queries as well. For this, we use the set of three-way positive conjunction queries as an example, i.e., $Q_3$. We compare the error against the answers returned from (independent) Laplace mechanism by choosing an appropriate value of $\epsilon'$ for each query in $Q_3$ according to the advanced composition theorem (see Section~\ref{sub:data-and-params}), such that overall $\epsilon$ is just under $1$ for queries in the set $Q_3$ only (i.e., the Laplace mechanism does not compute answers to $Q_{12}$). The results are shown in Figure~\ref{fig:threeway}. Once again our method outperforms the Laplace mechanism for the majority of the queries in $Q_3$ for all three datasets. Looking closely, we see from Table~\ref{tab:alpha-values}, that Lap actually performs better in terms of maximum absolute error for 1\% of the queries in the Adult dataset. However, for majority of the queries, $> 99\%$, over method performs better. In fact, our method outperforms Lap in terms of the 95\% and 99\% error profiles for all three datasets. For the {\dss} dataset the maximum error over all queries is similar to Lap, whereas for the Hospital dataset we again outperform Lap.


\begin{figure*}[!tbh]
\centering
\subfigure[Adult]{\label{fig:adult:threeway}
\includegraphics[width=0.31\textwidth]{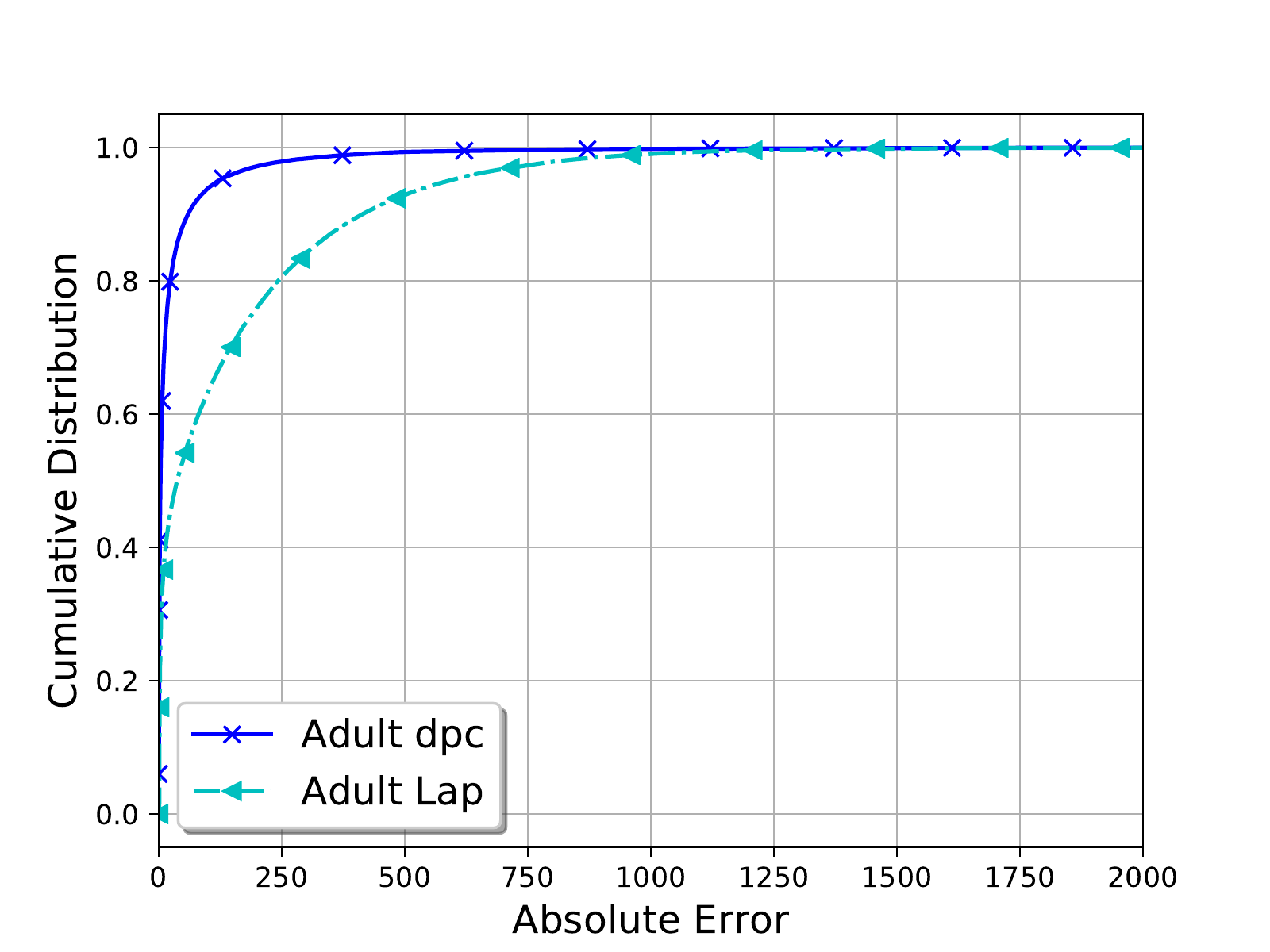}}
\subfigure[{\dss}]{\label{fig:dss:threeway}
\includegraphics[width=0.31\textwidth]{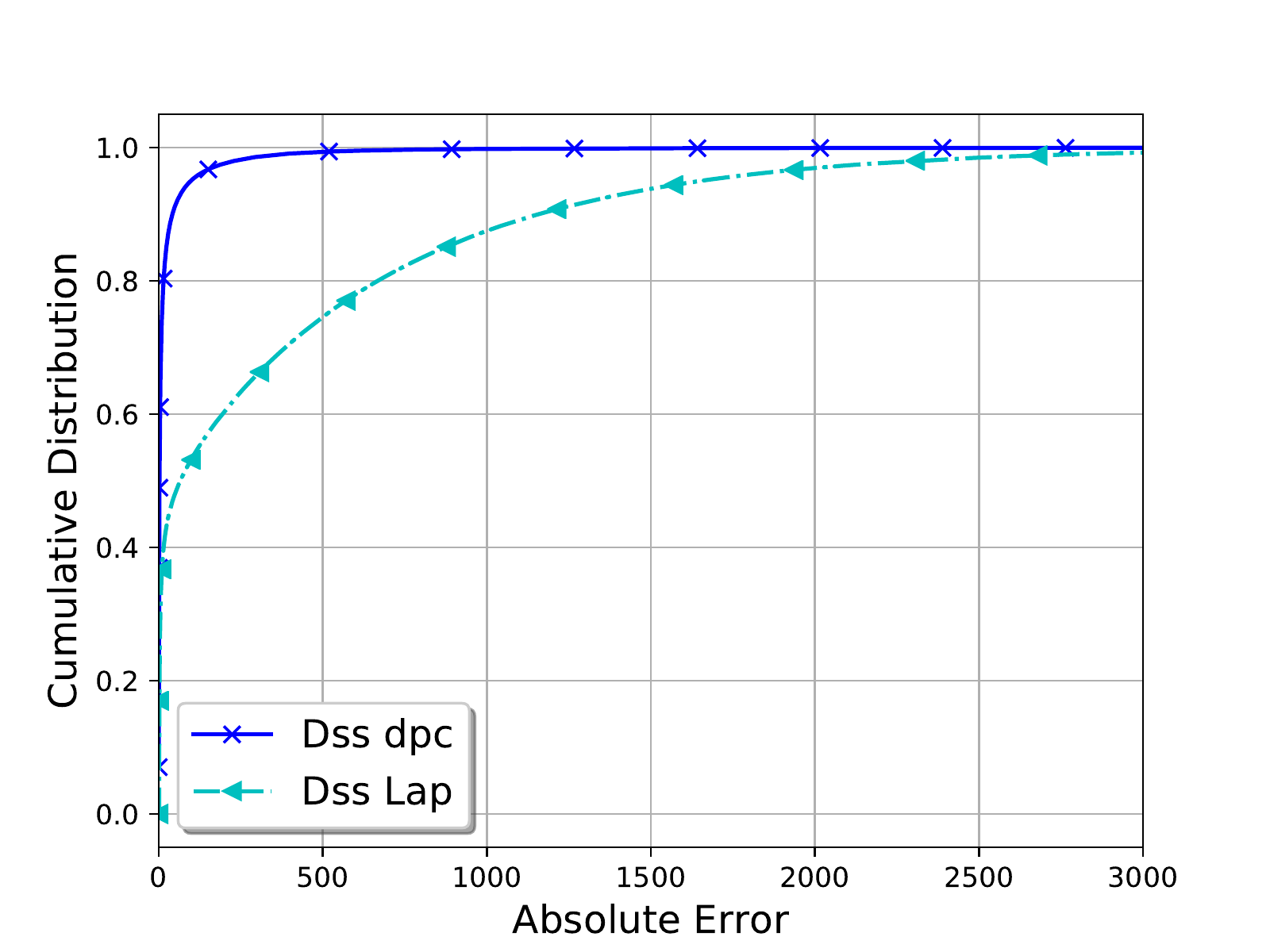}}
\subfigure[{\dss}]{\label{fig:hos:threeway}
\includegraphics[width=0.31\textwidth]{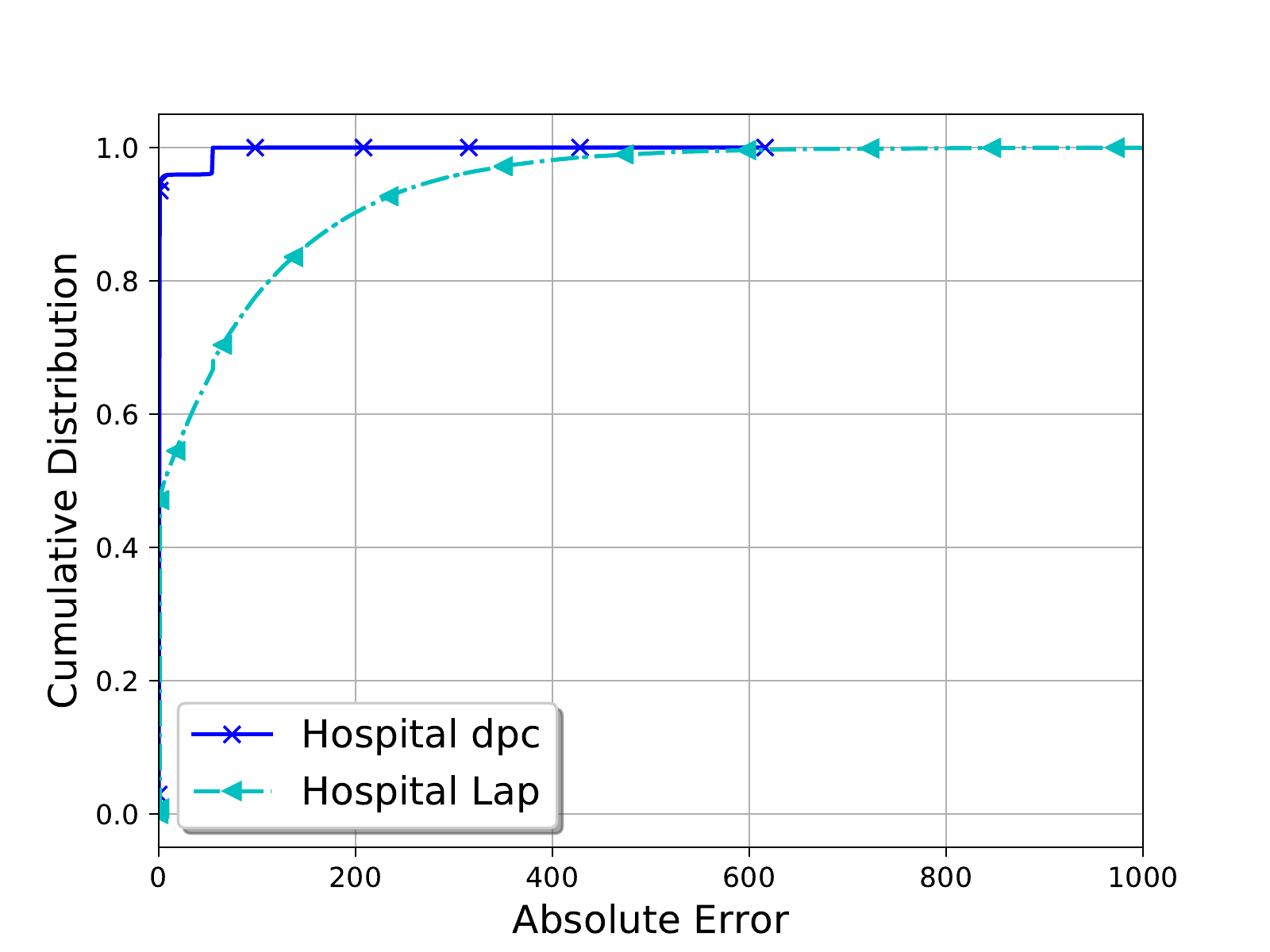}}
\caption{CDFs of the absolute error from our method against the Laplace mechanism on the set $Q_3$ of three-way positive conjunctions on the three datasets.}
\label{fig:threeway}
\end{figure*}

\subsubsection{Effect of the Privacy Parameter}
To show the effect of $\epsilon$ on utility, we vary it from $0.25$ to $5$ and report the error on the set $Q_{12}$. For this, we only use the Adult dataset as the effect is similar on the other two datasets. Figure~\ref{fig:eps} shows the CDF of the absolute error against different values of $\epsilon$. Notice that this is the overall privacy budget. With $\epsilon = 0.25$ we have average and maximum absolute errors of 357 and 3419, respectively, for the set $Q_1$, and 68 and 6421, respectively, for the set $Q_2$. With $\epsilon = 5$, the average and maximum absolute errora are much lower at $41$ and $179$, respectively, for $Q_1$, and 27 and 5882, respectively, for the set $Q_2$. As expected, the error profiles gradually improve as we move from $\epsilon = 0.25$ to $\epsilon = 5$. The error profiles are much similar for the set $Q_2$ then $Q_1$. As discussed in Section~\ref{subsub:cop-no-dp}, since the majority of attributes are uncorrelated, this implies that our method maintains that aspect by not adding too much noise on the set $Q_2$. 

\begin{figure*}[!tbh]
\centering
\subfigure[Set $Q_1$]{\label{fig:eps:adult:oneway}
\includegraphics[width=0.43\textwidth]{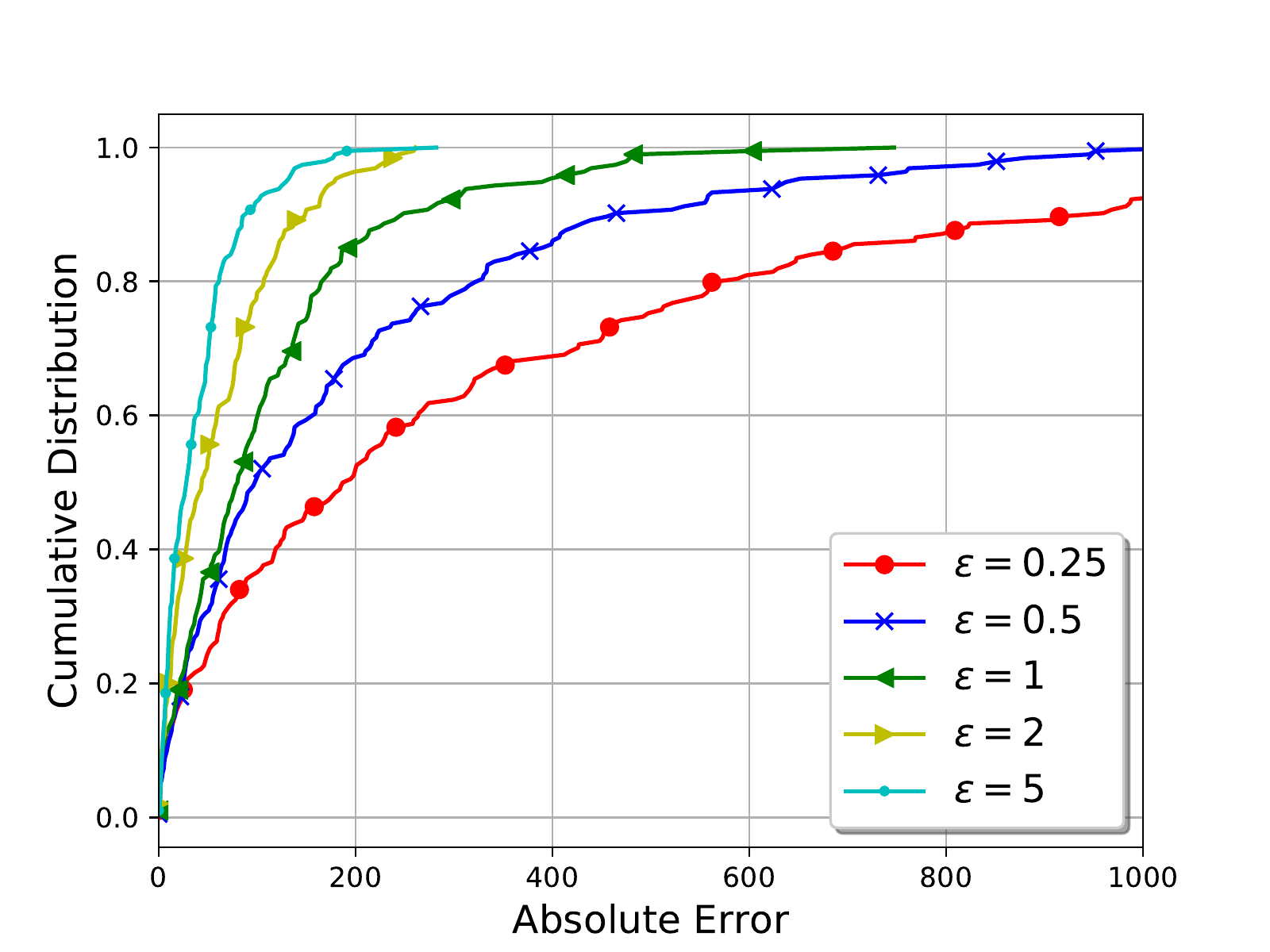}}
\subfigure[Set $Q_2$]{\label{fig:eps:adult:twoway}
\includegraphics[width=0.43\textwidth]{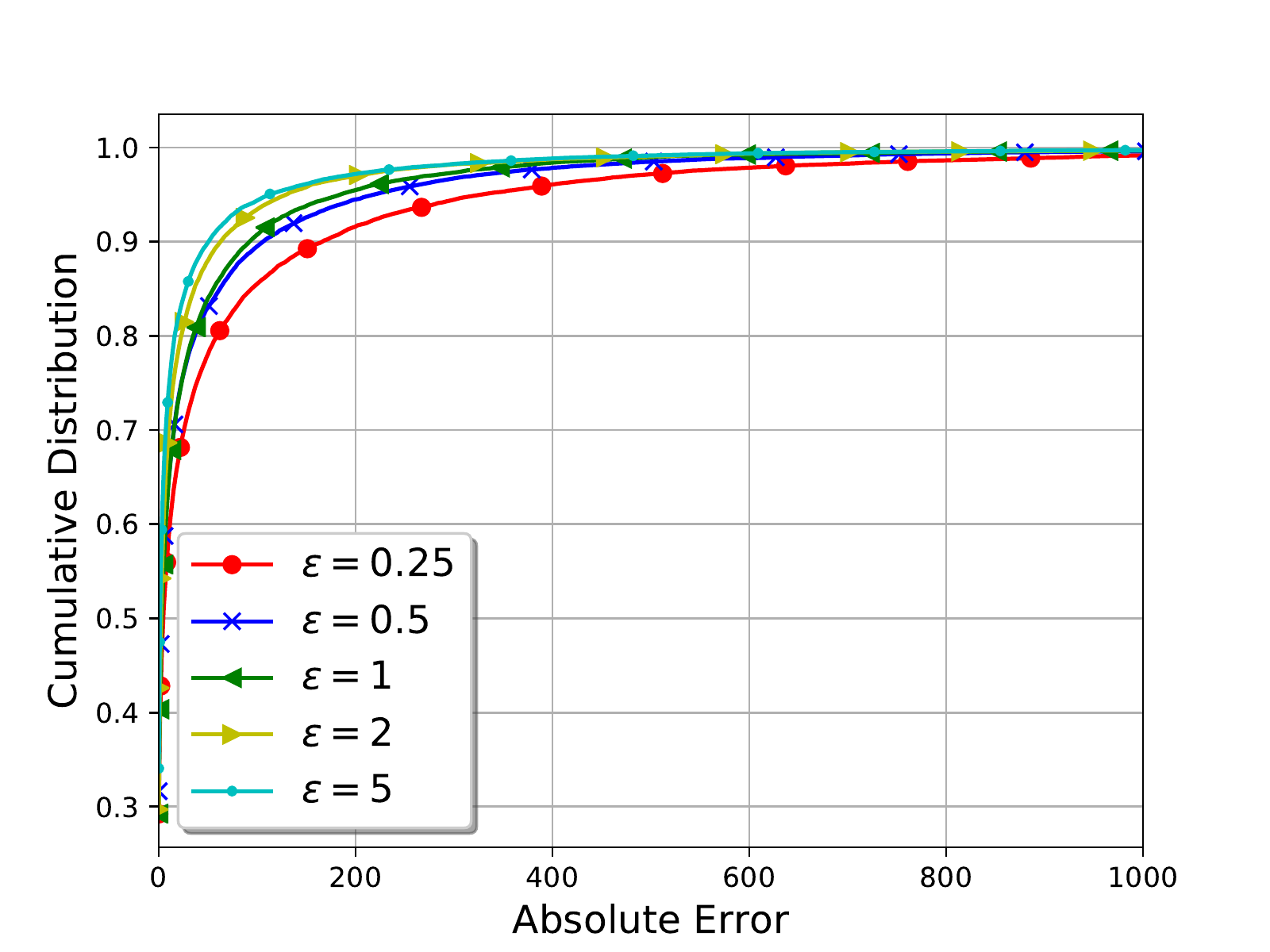}}
\caption{Absolute error on the set $Q_{12}$ against different values of $\epsilon$ on the Adult dataset.}
\label{fig:eps}
\end{figure*}

\subsubsection{Computational Time and Parallelism}
One of the motivations for using the proposed approach is its computational feasibility when the input data is high dimensional. 
We first note that our method is highly parallelizable. In particular, the computation of one-way and two-way positive conjunctions can be done in parallel. For the one-way marginals, parallel computation is straightforward. For the two-way conjunctions, we take the $i$th attribute (in the original dataset) and compute its conjunction with all attributes numbered $i + 1$ to $m$ attributes, for all $i \in [1, m - 1]$, assigning a separate process for each $i$. Obviously, the number of combinations for the first attribute is the highest, and becomes progressively less for latter attributes. Likewise, we also parallelize the computation of $\binom{d}{2} - d$ Gaussian correlations $\rho_{i, j}$ which uses a bisection search. While other components of our method can also be executed in parallel, e.g., generating synthetic records through the copula, we do not do so as these processes did not consume much computational time.

To generate the synthetic datasets we used a single-CPU Intel Xeon E5-2660 2.6GHz server with 10 cores and 128GB memory. Our implementation was done in \texttt{Python}.\footnote{\url{https://www.python.org/}} We parallelized part of our mechanism, as described above. The average run-times (over 10 runs) for the three datasets Adult, {\dss} and Hospital, are shown in Table~\ref{tab:run-times}. Obviously, the run-time is a function of the parameters $m$ (number of original attributes), $d$ (number of binary attributes) and $n$ (the number rows in the dataset). Asymptotically, the run-time of our method is $O(m^2 n + d^{2.38})$. The {\dss} dataset takes the longest time, which is understandable since it is about 3 times bigger in terms of the number of binary attributes and has 150 times more rows than the Adult dataset. If the number of rows is not large, then the run-time is not severely impacted by an increase in the number of binary attributes, as is indicated by the run-times of the Hospital dataset. To further assess the scalability of our algorithm, we constructed two synthetic (fake) datasets Syn 1 and Syn 2 with 1,000 and 2,000 binary attributes, respectively. Both had the same number of rows as the {\dss} dataset. The run times of these two datasets are shown under Syn 1 and Syn 2 in Table~\ref{tab:run-times}. By far, Syn 2 is the largest dataset, and even with this dataset we can generate a private synthetic dataset via our method in around 19 hours. We stress that since synthetic datasets need only be produced once, these times are practical. Thus, our method can output a privacy-preserving synthetic datasets of high dimensional datasets in reasonable time.

\begin{table*}
\centering
\begin{tabular}{c|c|c|c|c|c}
\multirow{2}{*}{Mechanism} & \multirow{2}{*}{Rows} & \multicolumn{2}{c|}{Attributes} & \multirow{2}{*}{Ave. Time} & \multirow{2}{*}{Runs} \\
\cline{3-4}
& & Original & Binary & & \\
\hline\hline
Adult & 32,560 & 14 & 194 & 6 min 47 sec ($\pm 32$ secs) & 10\\
{\dss} & 5,240,260 & 27 & 674 & 2 h 30 min ($\pm 15$ mins) & 10\\
Hospital & 10,000 & 9 & 1,201 & 46 min 21 sec ($\pm 11$ mins) & 10\\
Syn 1 & 5,240,260 & 1,000 & 1,000 & 5 h 43 min ($\pm 4$ mins) & 2\\
Syn 2 & 5,240,260 & 2,000 & 2,000 & 18 h 32 min ($\pm 12$ mins) & 2
\end{tabular}
\caption{Run-time.}
\label{tab:run-times}
\end{table*}

\section{Related Work}
\label{sec:rw}
In line with the theme of the paper, we restrict our review of related work to proposals for generating differentially private synthetic datasets. We divide this into two main categories. The first consists of mechanisms that provide \emph{provable} utility guarantees. The second is a class of algorithms that claims high utility in practice possibly relying on assumptions on the distribution of the input dataset, which we call heuristic approaches. Our method lies in this class. We review the two classes in order.

One way to release a synthetic dataset is to add (independent) Laplace noise to all point functions (Definition~\ref{def:pt}) from the input domain $\domain$~\cite{dp-book}. The resulting dataset gives good answers to point functions but lower order margins are noisier. However, the main problem with this approach is that its runtime is $O(| \domain |)$ which is exponential in the number of attributes; hence, its inapplicability to high dimensional datasets.
The stability-based histogram algorithm~\cite{bun-stable, balcer, salil-tut} runs in time only $O(\log |\domain| )$ by using the notion of local sensitivity and relying on approximate differential privacy. However, for high dimensional datasets it is likely that the output synthetic datasets will only contain a fraction of the original point functions (Definition~\ref{def:pt}), due to a high percentage of rows being unique or having low multiplicity in a high dimensional dataset.

For a more general class of counting queries, i.e., not necessarily point functions, the BLR algorithm~\cite{blr} and the MWEM algorithm~\cite{mwem} allow answers to exponentially many queries with noise per query proportional to $n^{2/3}$ and $n^{1/2}$, respectively ($n$ being the number of rows). However, these algorithms are not efficient as both require time polynomial in $| \domain |$. This makes these algorithms inefficient for high dimensional datasets. The drawback of exponential runtime (in the number of attributes) is also present in the mechanism from~\cite{dp-hard}, the median mechanism~\cite{median}, and the matrix mechanism~\cite{matrix} to name a few. For instance, Privlet~\cite{wavelet}, which can be categorised as an instance of the matrix mechanism, is designed to answer range queries by first creating a full contingency table (frequency matrix) of the input datasets. This is obviously exponential in the number of attributes of the dataset.

Computational inefficiency is not surprising since any synthetic data generation mechanism that answers an arbitrary number of counting queries, or even the set of all two-way marginals, is expected to run in exponential time under the hardness assumption of some well known cryptographic primitives~\cite{pcp-hard, ullman-n2}. However, algorithms that run in exponential-time in theory, might still be efficient in practice. The DualQuery algorithm~\cite{dual-query} is one such algorithm, which approximates a set of given counting queries, say three-way marginals, to within $n^{2/3}$ (absolute) error. The algorithm requires solving an optimization problem, which is hard in theory but solvable in practice for large parameters using standard optimization software. Likewise, the MWEM algorithm can run in reasonable time for a large number of attributes (up to 77 binary attributes) in practice~\cite{mwem}. Both approaches suggest further improvement in run-time using heuristics.

This leads us to the heuristic approaches for synthetic data release. Unlike the above mentioned class of algorithms, this class does not provide a provable utility guarantee and is often accompanied with some heuristic assumption on the input data distribution; crucially, for utility guarantees and not for privacy. As long as the heuristics hold true, the algorithm is expected to produce a synthetic dataset with good utility. The private spatial decomposition technique~\cite{psd} decomposes the input dataset into a hierarchical tree and then answers range queries over this structure. The technique is relevant to spatial data, and does not seem generic enough to consider categorical variables. As we discussed earlier, this requires fixing an artificial order on categorical variables which can be completely arbitrary.
PrivBayes is another algorithm~\cite{privbayes} which constructs a Bayesian network of an input dataset. The Bayesian network maintains attribute correlations and allows to approximate the data distribution as a set of low dimensional marginals. Efficiency is guaranteed so long as the degree of the network is low, where degree is roughly defined as the maximum number of attributes a given attribute depends on in the Bayesian network. The obvious assumption is that most correlations in the input datasets are of low degree. DiffGen~\cite{diffgen} proposes a generalization based approach for releasing data where a hierarchical tree is first constructed and a table corresponding to a given generalization level (in the tree) is released where the generalization level itself is decided by maximising utility through the exponential mechanism~\cite{dp-book}. This implies that the level of generalization of the output data is randomized, thus resulting in different utility on each invocation. This can be a drawback from a usability point-of-view if two datasets on the same domain but, say, different time periods are to be released, each resulting in a different level of generalization. The algorithm also runs in time exponential in the number of attributes.

\subsection*{Comparison with DPCopula}
The closest work to ours is that of Li, Xiong and Jiang~\cite{dp-copula} who propose DPCopula. DPCopula also uses the Gaussian copula to generate differentially private synthetic datasets. However, there are considerable differences between our work and theirs. In essence, we claim that our method is more general and efficient as detailed by the following four major differences.

\subsubsection*{Categorical Attributes.}
DPCopula imposes an order on the values of any
categorical (nominal) attributes in the input data set. However, we argue that this order is inherently artificial and arbitrary: different choices of order for categorical attributes 
produces different pair-wise correlations between them. This in turn effects the accuracy of two-way conjunctions computed on data generated through the Gaussian copula using these pair-wise correlations. 
A simple example illustrates our point. Consider a database having two attributes $X$ (``country of birth'') and $Y$ (``marital status'') with possible values (\texttt{English}, \texttt{Chinese}, \texttt{French}) and (\texttt{Married}, \texttt{Divorced}, \texttt{Widowed}), respectively. 
For the sake of simplicity, assume that the dataset consists of 400 records with 100 pairs of (\texttt{English}, \texttt{Married}), 200 pairs of \texttt{Chinese}, \texttt{Divorced}) and 100 pairs of (\texttt{French}, \texttt{Widowed}). To calculate Pearson correlation between $X$ and $Y$, let us fix the numerical map $(\texttt{English}, \texttt{Chinese}, \texttt{French}) \rightarrow (1, 2, 3)$ on attribute $X$. Consider first the numerical map $(\texttt{Married}, \texttt{Divorced}, \texttt{Widowed}) \rightarrow (1, 2, 3)$ on attribute $Y$. The Pearson correlation between $X$ and $Y$ in this case is exactly 1. However, notice that there is no logical reason to choose any of the two maps. If we change the second map to the (equally valid) map $(\texttt{Married}, \texttt{Divorced}, \texttt{Widowed}) \rightarrow (3, 1, 2)$, the resulting correlation becomes $\approx -0.457$. If we use the resulting correlations to generate synthetic outputs via the Gaussian copula,  we obtain drastically different results on the two-way counts. A simple program in \texttt{R} results\footnote{This is done using the \texttt{rCopula} function from the \texttt{copula} package for \texttt{R}~\cite{rcopula}.} in the two-way counts $\#(\texttt{English}, \texttt{Married}) = 112$, $\#(\texttt{Chinese}, \texttt{Divorced}) = 193$, and $\#(\texttt{French}, \texttt{Widowed}) = 95$ for the first map. The second map results in the counts $\#(\texttt{English}, \texttt{Married}) = 47$, $\#(\texttt{Chinese}, \texttt{Divorced}) = 35$ and $\#(\texttt{French}, \texttt{Widowed}) = 46$, from the synthetic output. This simple example illustrates the impact of an arbitrary order on correlations between categorical attributes in the original data set. While the first ordering gives good results, the second ordering gives noticeably bad results. The reason why the first ordering gives good results is mainly an artefact of the simplicity of illustration. With more attributes, where multiple inter-attribute correlations need to be determined, a utility maximizing ordering across all categorical attributes may not be straightforward. Appendix~\ref{app:cor-art-ord} gives a more analytical treatment on the impact of changing orders (maps) on the correlation. As opposed to DPCopula, our proposed method (Section~\ref{sec:method}) does
not rely on arbitrary orders for nominal attributes. Thus our method is
capable of producing synthetic datasets with pair-wise attribute correlations
that are close to the ones in the original dataset.

\subsubsection*{Small Domain Attributes.} DPCopula is designed only for attributes with large domains, i.e., attributes which have at least 10 different values~\cite[\S 4.4]{dp-copula}. For small domain (including binary attributes) a method called DPHybrid is proposed in~\cite{dp-copula} which partitions the data into smaller datasets (one per attribute value in the small domain attributes) and then generates separate synthetic datasets per partition using Gaussian copulas, before eventually combining them. First, if the dataset has only small domain attributes then DPCopula or its hybrid variant cannot be used. Secondly, depending on the number of small domain attributes DPHybrid can become computationally infeasible, i.e., taking time exponential in the number of small domain attributes. For instance, in our DSS dataset, we have a total of 10 small domain attributes (having number of values less than 10) totalling approximately $2^{18}$ partitions (product of attribute values). Thus, the time to produce the combined synthetic dataset is $2^{18}$ times the time to produce individual synthetic datasets via the Gaussian copula for each partition; which itself takes time $O({m'}^2 n)$, where $m'$ is the number of large domain attributes (17 in the DSS dataset). This amounts to roughly $2^{50}$ time to generate a synthetic dataset from the DSS dataset. With more small domain attributes, this is bound to increase.

\subsubsection*{Correlation Matrix.} A third major difference between our work and DPCopula is in the process to generate the differentially private correlation matrix, i.e., $\mathbf{P}'$ in Eq.~\ref{eq:Q2}. There are two methods described in~\cite{dp-copula} to generate the counterpart to $\mathbf{P}'$. The first method uses Kendall's rank correlation coefficient $\tau$~\cite{kendall} to measure correlations between attributes in the original dataset and then uses the relation $\mathbb{E}(\tau) = \frac{2}{\pi} \sin \rho$ to obtain the correlation coefficient $\rho$ between the corresponding normal random variables. A differentially private variant is constructed by showing that $\tau$ has low global sensitivity~\cite{dp-copula}. We first note that the relation $\mathbb{E}(\tau) = \frac{2}{\pi} \sin \rho$ is proven for continuous random variables~\cite[\S 3.2]{spear-ken},\cite{sgn-expansion}. This is one reason why DPCopula is targeted for continuous data or at least large domain discrete attributes (approximated as continuous attributes). Secondly Kendall's rank correlation coefficient, as the name suggests, assumes an order between attributes; once again, as argued before, for categorical attributes this means that an artificial order needs to be induced which is not reflective of the correlations. Since we convert data into a binary format, there is no meaningful rank between two binary variables that could be used to compute Kendall's tau coefficient. Furthermore, the conversion $\tau = \frac{2}{\pi} \sin \rho$ would not apply as well. The second method used by DPCopula is a maximum likelihood estimation method to compute $\rho$ using a similar ``sample-and-aggregate'' method described in~\cite{what-we-know}. This method involves partitioning the dataset into $n/l$ partitions and then adding Laplace noise of scale ${2 \binom{m}{2}}/{l \epsilon}$ to each of the $\binom{m}{2}$ pairs of attributes. Since our data is in binary format, we would need to add noise of scale ${2 \binom{d}{2}}/{l \epsilon}$. If we do not want the noise to overwhelm the calculation of $\rho$, we need $l$ to be at least $\binom{d}{2}$. Unfortunately, this means that we would have the partitions of size much smaller than $\sqrt{n}$ for all three datasets considered in this paper, which is needed for a good approximation of $\rho$'s~\cite[\S 3.1.2, p. 145]{what-we-know}. We therefore use a different method for constructing the differentially private correlation matrix by adding noise to the margins before obtaining Pearson product-moment correlations and then using a bisection search to convert to corresponding correlations for Gaussian variables.

\subsubsection*{Positive Definite Matrix.}
In order for Cholesky decomposition to work (cf. Section~\ref{subsec:algo_gen_multivariate_norm_distri}), DPCopula uses a heuristic method for obtaining a positive definite matrix with unitary diagonals~\cite{rousseeuw, dp-copula}. Like our method, the procedure first finds the eigen decomposition of the matrix $\mathbf{P}$, i.e., the matrix of Gaussian correlations obtained from differentially private correlation matrix of input data (see Section~\ref{subsec:algo_orig2Gauss_mapping}), and fixes the negative eigenvalues to a small value or the absolute value. The difference from our method is that to make the resulting matrix into a correlation matrix they normalize the matrix. The resulting procedure does indeed return ``a'' correlation matrix. However, this heuristic step does not guarantee that the resulting matrix is the nearest correlation matrix to the input matrix $\mathbf{P}$. By using the algorithm from~\cite{higham2002computing}, we arrive at the nearest correlation matrix to the given matrix as defined by the given matrix norm. 

\section{Conclusion}
\label{sec:conclude}
We have presented a generic mechanism to efficiently output differentially private synthetic datasets with high utility using the concept of Gaussian copulas. Our method is generic; while Gaussian copulas are mostly used to generate (non-private) synthetic datasets for numerical attributes, our methods is applicable to both numerical and categorical attributes alike. The proposed mechanism is efficient as it takes time polynomial in the number of attributes, in contrast to exponential time required by many differentially private synthetic data generation algorithms, which makes our algorithm suitable for high-dimensional datasets. Through experiments on three real-world datasets, we have shown that our mechanism provides high utility, matching and even surpassing the utility provided by independent noise through the Laplace distribution. A shortcoming of our work is the lack of a provable utility guarantee. Nonetheless, we have provided significant experimental evidence of utility. A future direction is to provide theoretical guarantees of utility, perhaps by assuming certain characteristics of the distribution of the input dataset which may make the analysis tractable. A further interesting direction is to assess if other copulas found in literature could also be used to efficiently generate synthetic datasets with high utility.



\bibliographystyle{abbrvnat}
\bibliography{draft2}

\appendix

\section{Artificial Order Disrupts Pearson Correlation}
\label{app:cor-art-ord}
Consider two categorical attributes taking $n$ values each over a database of size $n$. Let us assign the sequential order $1, \ldots, n$ to the $n$ values of the first attribute (say based on lexicographical order). Let us define another order on the second attribute in which the order of the first $\lambda n$ values are reversed. The remaining $n - \lambda n$ values retain the sequential order, where $\lambda \in [0, 1]$. For instance, $\{3, 2, 1, 4, 5\}$ is the order on the second attribute with $\lambda = 0.6$, i.e., the first 3 values have a reverse order. Note that the mean $\mu$ is the same for both orders, given by $\mu = \frac{n+1}{2}$. Let $r_{\lambda}$ be the correlation coefficient between the two attributes, and let $y_i$ denote the $i$th value in the second attribute. Then
\begin{equation}
\label{eq:cor-ord}
r_\lambda = \frac{\sum_{i = 1}^n (i - \mu)(y_i - \mu)}{\sqrt{\sum_{i = 1}^n (i - \mu)^2 \sum_{i = 1}^n (y_i - \mu)^2}}.
\end{equation}
Now consider the denominator in the above. After simplification, we get
\begin{align}
\text{den} &= \sqrt{\sum_{i = 1}^n (i - \mu)^2 \sum_{i = 1}^n (y_i - \mu)^2} \nonumber\\
					&= \sqrt{\sum_{i = 1}^n (i - \mu)^2 \sum_{i = 1}^n (i - \mu)^2} \nonumber\\
					&= \frac{n(n+1)(n-1)}{12}. \label{eq:denom}
\end{align}
Consider now the numerator, which after simplification gives
\begin{equation}
\text{num} = \sum_{i = 1}^n iy_i - \frac{n(n+1)^2}{4}. \label{eq:num}
\end{equation}
Now let $r_0$ be the correlation when the two attributes have the same order, i.e., $\lambda = 0$. We are interested in finding $r_0 - r_\lambda$ as a function of $\lambda$, where different values of $\lambda$ indicate the level of change in the order. Using the fact that $y_i = i$, when $\lambda = 0$, through Eqs.~\ref{eq:denom} and~\ref{eq:num} we obtain
\begin{align}
	r_0 - r_\lambda &= \frac{1}{\text{den}} \cdot \sum_{i = 1}^n (y_i - i) i \nonumber\\
								&= \frac{1}{\text{den}} \cdot \sum_{i = 1}^{\lambda n} (\lambda n - i + 1 - i) i  \nonumber \\
								&= \frac{1}{\text{den}} \left( (\lambda n + 1) \left( \sum_{i = 1}^{\lambda n} i \right) - 2 \left( \sum_{i = 1}^{\lambda n} i^2 \right) \right) \nonumber \\
								&=  \frac{1}{\text{den}} \frac{\lambda n (\lambda n + 1) (\lambda n  - 1)}{6} \nonumber \\
								 &= \frac{12}{n(n+1)(n-1)}  \frac{\lambda n (\lambda n + 1) (\lambda n - 1)}{6} \nonumber \\
													&= 2\lambda^3 \left( 1 - \frac{\lambda^{-1} - 1}{n+1} \right) \left( 1 - \frac{\lambda^{-1} - 1}{n - 1} \right), \label{eq:cor-diff}
\end{align}
where $\lambda \ne 0$ in the last equality. Thus, for instance if $\lambda = 0.5$, we get $r_0 - r_{0.5}  \approx \frac{1}{4}$. And when $\lambda = 1$, i.e., complete reversal of order, we get the difference as $2$. Since $r_0 = 1$, this means that $r_1 = -1$, a complete reversal in correlation.

\section{Proof of Lemma~\ref{lem:perturb_CDF}}
\label{app:lem:perturb_CDF}
Note that conversion of $A$ into $X_j$'s creates $|A|$ distinct partitions of the domain $\domain$. From the parallel composition theorem, i.e., Theorem~\ref{the:par-comp}, since each marginal is computed with $(\epsilon'_i, 0)$-differential privacy, the overall differential privacy guarantee remains $(\epsilon'_i, 0)$.

Another way of looking at this is as follows. Suppose, instead of converting $A$ into binary attributes, we compute its marginal distribution directly from the histogram of the values $A$ takes, where each histogram bin corresponds to the number of occurrences of a unique value of attribute $A$. Since each row of $D$ can only be in one of the bins, the private version of the histogram can be obtained by adding Laplace noise of scale $2/\epsilon'_i$ to each count, and then publishing the counts. The resulting mechanism remains $(\epsilon'_i, 0)$-differentially private~\cite[\S 3.3, p. 33]{dp-book}. Now, we can convert $A$ to binary attributes and deduce the marginals of these binary attributes from the histogram counts. This does not further impact privacy, as it is simply post-processing (See Theorem~\ref{the:post-proc}).

\section{Pearson Correlation over Binary Attributes has High Global Sensitivity}
\label{app:cor-high-sen}
Consider an $n$-row binary database $D_1$ having two attributes $X$ and $Y$ with only the first entry in each attribute set to 1 and the rest to 0, i.e., $X_1 = Y_1 = 1$ and $X_i = Y_i = 0$ for all $i \in \{2, \ldots, n\}$. Consider the neighbouring database $D_2$ which is the same as $D_1$ except that $X_2 = 1$. Let $r_1$ be the correlation coefficient between $x$ and $y$ in $D_1$, and let $r_2$ be its counterpart in $D_2$. We will show that $r_1 - r_2$ is large, meaning that the correlation coefficient has high global sensitivity and any noise scaled to the sensitivity of the correlation coefficient will overwhelm the accuracy of the results. Let $\overline{X}$ and $\overline{Y}$ denote the mean of the attributes $X$ and $Y$, respectively. We have
\begin{align}
r_1 &= \frac{\sum_{i = 1}^n X_i Y_i  - n \overline{X} \cdot \overline{Y}}{ \sqrt{\sum_{i = 1}^n X_i^2 - n\overline{X}^2 }  \sqrt{\sum_{i = 1}^n Y_i^2 - n\overline{Y}^2 }} \nonumber\\
& = \frac{1 - n \cdot \frac{1}{n}\frac{1}{n}}{\sqrt{1 - n \cdot \frac{1}{n^2}} \sqrt{1 - n \cdot \frac{1}{n^2}}} \nonumber\\
& = \sqrt{1 - \frac{1}{n}}. \label{eq:gs:r1}
\end{align}
Similarly,
\begin{align}
r_2 &= \frac{1 - \frac{2}{n}}{\sqrt{2 - \frac{4}{n}} \sqrt{1 - \frac{1}{n}}} \nonumber\\
& = \frac{\sqrt{1 - \frac{2}{n}}}{\sqrt{1 - \frac{1}{n}}} \times \frac{1}{\sqrt{2}}. \label{eq:gs:r2}
\end{align}
From Eqs.~\ref{eq:gs:r1} and~\ref{eq:gs:r2}, with large enough $n$, we get
\begin{equation*}
r_1 - r_2 \approx 1 - \frac{1}{\sqrt{2}} \approx 0.29
\end{equation*}
Thus, the global sensitivity of the Pearson product moment correlation for binary attributes is at least 0.29. Adding Laplace noise scaled to this will substantially alter the correlation between attributes and hence the corresponding counts.

\section{Computing Two Way Margins in \texorpdfstring{$O(m^2 n)$}{O(m2n)} time}
\label{app:faster-two-way}
A straightforward way to compute all two-way margins over the binary dataset is as follows: for each pair of binary attributes $(B_1, B_2)$ scan the dataset of $n$ rows and record the number of times each possible value $(b_1, b_2)$ occurs, where $b_1, b_2 \in \{0, 1\}$. However, this takes time proportional to $O(d^2 n)$, which can be prohibitive if $d$ is large. We will show a method below that requires time only $O(m^2 n + d^2)$. Recall that $m$ is the number of attributes in the original dataset $D$, and $d$ in its binary expansion $D_{\mathsf{B}}$.

First we compute one-way margins for an attribute $A$ by scanning the database and creating a new hash entry for any new entry $a \in A$ and updating its count in the hash table, all in $O(n)$ time. We then go through each element $a$ in the hash table for $A$, letting $b = \text{bin}(a)$ be its binary representation, and creating the corresponding binary version of the hash entry $\text{hist}_{\mathsf{B}}(b)$. This can be done in $O(dn)$ time.

Now for each possible pairs of attributes $A_1, A_2$ in the database $D$, we initialise another hash table: ``$\text{hist}$.'' If we see a new pair of values $(a_1, a_2)$, we create the entry $\text{hist}(a_1, a_2)$ and set it to 1. Otherwise we increment the counter. Now for each existing value $(a_1, a_2)$, we set $b_1 = \text{bin}(a_1)$ and $b_2 = \text{bin}(a_2)$. Let $b_1b_2$, $b_1\overline{b}_2$, $\overline{b}_1b_2$ and $\overline{b}_1\overline{b}_2$ denote the number of occurrences of $(1, 1)$, $(1, 0)$, $(0, 1)$ and $(0, 0)$, respectively. Then, these can be computed as
\begin{align*}
	b_1b_2 &= \text{hist}(a_1, a_2), \\
	b_1\overline{b}_2 &=  \text{hist}_{\mathsf{B}}(b_1) - \text{hist}(a_1, a_2), \\
	\overline{b}_1b_2 &= \text{hist}_{\mathsf{B}}(b_2) - \text{hist}(a_1, a_2), \\
	\overline{b}_1\overline{b}_2 &= n - b_1b_2 - b_1\overline{b}_2 - \overline{b}_1b_2.
\end{align*}
It is easy to see that the above can be computed in $O(m^2 n + d^2)$ time.

\end{document}